\documentclass{article}
\newtheorem{Assumption}{Assumption}
\newtheorem{Definition}{Definition}
\newtheorem{lemma}{lemma}
\newtheorem{remark}{remark}
\newtheorem{theorem}{theorem}
\newtheorem{proof}{proof}
\newtheorem{case}{case}

\usepackage{PRIMEarxiv}
\usepackage{amsmath}
\usepackage{amssymb}
\usepackage{CJKutf8}
\usepackage[utf8]{inputenc} 
\usepackage[T1]{fontenc}    
\usepackage{hyperref}       
\usepackage{url}            
\usepackage{booktabs}       
\usepackage{amsfonts}       
\usepackage{nicefrac}       
\usepackage{microtype}      
\usepackage{lipsum}
\usepackage{fancyhdr}       
\usepackage{graphicx} 
\graphicspath{{media/}}     

\title{Fault-tolerant control of random switching topology multi-agent system based on event triggering

}

\author{
  Ouyang Lingcong\quad\quad   Kaijun Yang \\
  School of Electrical and Control Engineering, Shaanxi University of Science and Technology\\
  Xi'an \quad  China\\
  \texttt{\{Kaijun Yang\}yangkj2020@126.com} \\
}

\begin{document}
\maketitle

\begin{abstract}
In this paper, the formation control of multi-agent systems in random switching communication topology is studied, and the problem of excessive bandwidth and low control efficiency among multi-agents is solved. For nonlinear multi-agent systems, a sliding mode formation control algorithm with event-triggered random switching communication topology is proposed. Firstly, a fault-tolerant control strategy based on stochastic system model is designed to solve the problem of low controller efficiency and controller failure during the formation of multi-agent systems. Compared with the traditional multi-agent system dynamic model, the stochastic system model has stronger universality, which improves the efficiency of the controller. Secondly, in order to deal with the problem of high communication load and frequency between agents during formation in random switching communication topology, the sample-based event triggering controller is introduced into the controller, which effectively reduces the communication frequency and reduces the impact of communication delay. Finally, the stability and robustness of event-triggered random switching topology formation control are verified by computer simulation.
\end{abstract}

\keywords{ multi-agent system\and formation control \and switching \and communication topology
\and   event-triggering control }

\section{INTRODUCTION}\label{sec1}
In multi-agent system formation, the traditional periodic communication method usually leads to a large amount of communication load, which affects the real-time and efficiency of the system. The sampling event trigger mechanism effectively reduces the communication frequency and communication load by only communicating when certain conditions are met, which is one of the core features of this study. Multi-agent systems often face model uncertainty and external interference in practical applications. With its inherent robustness, sliding mode control can effectively resist these uncertainties and disturbances, so as to ensure the stable operation of the system in a complex environment. In multi-agent systems, failure of communication links or other components is inevitable. By designing fault-tolerant control strategy, the system can maintain stable formation and normal operation in the face of these faults, and improve the reliability and stability of the system. The traditional static communication topology lacks adaptability and robustness and is difficult to cope with the dynamic changes of the environment. By introducing the dynamic communication topology triggered by sampling events, the system can flexibly adjust the communication structure according to environmental changes, thus improving the adaptability and robustness of the system.

The choice of communication topology affects the mode and efficiency of information exchange in multi-agent systems. The traditional communication topology is usually static, that is, the communication link between agents is fixed and does not change as the environment changes. However, static topologies often lack adaptability and robustness in the face of uncertainty and dynamic environments. In addition, static communication topologies are prone to communication link failures, which may lead to the failure of the entire multi-agent system. In order to deal with communication signal loss and ensure stable information exchange between multiple agents, it is an effective method to model the switching process of communication topology as Markov process. The consistency problem of first-order and second-order multi-agent systems under Markov switching communication topology has been discussed in recent studies\cite{shang2016consensus, mo2018mean}. At the same time, the problem of consistency tracking for discrete second-order multi-agent systems with Markov switching topology is also discussed \cite{xie2018group}. There are also a series of related studies on the consistency of continuous time and discrete time multi-agent systems \cite{li2015finite, sun2021consensus}, and obtain sufficient consistency conditions. However, most of the current research focuses on the consistency tracking during random switching, and there are relatively few researches on the formation tracking of multi-agent systems in random switching communication topology. Therefore, it is necessary to further study the formation tracking problem of multi-agent systems in random switching topology to ensure that multi-agent systems can effectively maintain stable formation and complete expected tasks in uncertain and dynamic environments.

In order to overcome the limitations of static communication topology, a dynamic communication topology method based on sampling event triggering is proposed. The method determines whether communication is needed by periodic sampling, and only triggers communication when the system state changes, thus reducing the communication overhead and improving the robustness of the system. In recent years, sampling event triggering control has made remarkable progress. By setting trigger conditions, this mechanism can exchange information only when a specific event occurs, effectively reducing communication frequency and resource consumption. For example, the adaptive trigger condition\cite{huang2019adaptive,gu2017adaptive,xu2020distributed} dynamically adjusts the threshold based on system status changes to ensure system stability and minimize the communication load. In addition, the mechanism has been successfully applied to UAV formation, mobile robot network and sensor network, which significantly improves the adaptability and robustness of the system in complex environments. In terms of theoretical research, the stability conditions of event-triggered control systems are proved through mathematical models such as Lyapunov function and inequality technique. In general, the sampling event triggering control method \cite{zhang2016overview,li2014event,fan2018sampling} can effectively reduce the communication frequency while ensuring the system performance, and significantly improve the communication efficiency and robustness of multi-agent systems.

In order to solve communication failures in multi-agent systems, the researchers focused on exploring various fault-tolerant control strategies aimed at enhancing the resilience of the system so that agents can flexibly adjust their behavior in response to communication interruptions. However, most of the existing research is based on complete information about the normal operation of dynamic systems, while in reality, system reliability is increasingly required due to potential component failures and system process anomalies. The fault-tolerant control strategy of the multi-agent system is discussed. Some of these studies combine hybrid trigger control \cite{ye2017adaptive,sakthivel2017fault, sathishkumar2020hybrid} and fault-tolerant sampled data control to deal with possible system failures. In addition, a robust fault-tolerant control method \cite{hao2013robust,ye2017fault} is also studied to deal with the fault situation in the system. In addition, the robust fault-tolerant control design of the system is studied based on the fuzzy model when the actuator fault changes with time. However, the sliding mode control method is rarely considered in the study of event-triggered consistency. Although the event-triggered strategy is introduced into the sliding mode control, the continuous information of the system cannot be used because the sampling replaces the continuous information, which challenges the robust consistency of quantitative description of multi-agent systems under the event-triggered sliding mode control strategy and stimulates the development of this work.

Sliding mode control (SMC), as an effective robust control method, can reduce the adverse effects of model uncertainty and unknown external interference in the control system. Due to its inherent robustness, sliding mode control is widely considered and combined with other control techniques in many theoretical and practical engineering systems\cite{zhao2017h}. In related research, adaptive sliding mode control strategy is used to study the stability of switching complex network systems \cite{plestan2010new,fei2020neural,li2018robust}. At the same time, a method combining dynamic event triggering control and sliding mode observer is proposed to derive the stability conditions of time-delayed T-S fuzzy systems\cite{liu2019observer,fang2019sliding}. In addition, an asynchronous method is used to ensure the stability of nonlinear Markov jump systems and improve their dissipative performance. In addition, a sliding mode control algorithm based on self-triggering \cite{song2021self,wang2020self,behera2015self} is proposed to analyze the stability of linear systems with external disturbances. The researchers also propose a new combined approach combining terminal sliding mode control and disturbance observe-based control to study the finite-time consistency problem of higher-order multi-agent systems. However, as far as we know, few people have studied the application of sliding mode control methods to scale consistency problems, let alone event-triggered sliding mode control methods. Therefore, this study aims to fill the gap and discuss the application of sliding mode control methods in scale consistency problems, especially the study of event-triggered sliding mode control methods \cite{wang2016distributed}.

To sum up, there are relatively few comprehensive studies on sampling event-triggered dynamic communication topology and sliding-mode formation control, which mainly focus on a single field and lack of integration and optimization of the two. Therefore, this paper aims to explore the integration of random switching communication topology triggered by sampling events with sliding mode formation control to achieve robust formation control of multi-agent systems in dynamic environments.
The contributions of this paper are as follows:

For stochastic switching topology nonlinear multi-agent systems, a stochastic switching topology multi-agent system sliding mode formation fault-tolerant control strategy is proposed based on a stochastic system model\cite{qi2019sliding}. This strategy ensures the stability of the multi-agent system. Compared to the traditional dynamic models used in multi-agent systems, the stochastic system model adopted in this paper is more versatile, and the controller used is more efficient.Addressing the issue of high communication load and frequency between agents in the formation process of multi-agent systems with stochastic switching communication topology, a sample point-based event-triggered controller is introduced into the sliding mode controller. This effectively reduces communication frequency and mitigates the impact of communication delays.

Then the rest of the paper is organised as follows. Some preliminaries, model formulation and some useful definitions are given in Section 2. Main results are proposed in Section 3. In Section 4, simulations are presented to illustrate the effectiveness of our results. Section 5 comes to a conclusion.

\section{PROBLEM STATEMENT AND PRELIMINARIES}\label{sec2}
\subsection{Graph theory}
The information exchange topology among a network of agents is modelled as a digraph $\mathcal{G}=(\mathcal{V}, \mathcal{E}, \mathcal{A})$, where $\mathcal{V}=\{1,2, \ldots, N\}$ and $\mathcal{E} \subseteq\{(j, i): j, i \in \mathcal{V}, j \neq i\}$ are the set of nodes and edges, respectively. $N_{i}=$ $\{j \in \mathcal{V}:(j, i) \in \mathcal{E}\}$ denotes the neighbour set of agent $i$. $\mathcal{A}=\left[a_{i j}\right]_{N \times N}$ represents the adjacency matrix with $a_{i j}>0$ if and only if $j \in N_{i}$, and $a_{i j}=0$ otherwise. The degree matrix of the graph $\mathcal{G}$ can be represented as $\mathcal{D}=\operatorname{diag}\left\{d_{i i}\right\} \in R^{N \times N}$ with $d_{i i}=\sum_{j=1}^{n} a_{i j}$. Correspondingly, the Laplacian matrix is denoted by $L=\mathcal{D}-\mathcal{A}$. A directed path from node $j$ to node $i$ is a sequence of ordered edges with the form $\left(j, i_{1}\right),\left(i_{1}, i_{2}\right), \ldots,\left(i_{p-1}, i_{p}\right),\left(i_{p}, i\right)$ where the nodes $i_{k} \in$ $\mathcal{V}, k=1,2, \ldots, p$, are distinct.A directed tree is a digraph, in which every node has exactly one parent except for the root defined as the only node which has no parent but a direct path to every other node. \\
A directed spanning tree is a directed tree, which consists of all the nodes and some edges in $\mathcal{G}$. A directed graph is said to contain a directed spanning tree if one of its subgraphs is a directed spanning tree. 
\subsection{Markovian jump process}
Let  $ G_{i} (t)=\textrm{Pr}( h_{n+1} \le t \mid R_{n} =i) $ 
which is the distribution function of the sojourn time when the mode stays in mode $i$. The transition probability of the embedded Markov process $\left\{R_n \right\}_{n \in N} $is defined by $q_{ij}= \textrm{Pr}(R_{n+1}= j \mid R_n= i)$, for any $i,j \in S ,i \neq j ,n\in N $.Since $G_i(t)$ only depends on the current mode $i$, it can be verified that $G_i(t)=\textrm{Pr}(h_{n+1} \le t\mid R_n=i)=\textrm{Pr}(h_{n+1}\le t \mid R_n=i,R_{n+1}=j )$ 
Define the transition rates $\pi_{ii}(h)=-\frac{g_i(h)}{1-G_i(h)}$ and $\pi_{ij}(h)=-\frac{q_{ij}g_i(h)}{1-G_i(h)},j\neq i $ where $g_i(h)$ is the probability density function of $G_i(h)$.The transition process of the Markov model is determined by the following transition rates:
\begin{equation}
	\begin{aligned}
		\textrm{Pr}(r(t+h)=j\mid r(t)=i)=\begin{cases}
			\pi_{ij}(h)h+o(h)\quad j\neq i, \\
			1+\pi_{ii}(h)h+o(h)\quad j= i.\end{cases}
	\end{aligned}
\end{equation}
\subsection{ Problem formulation}
We consider a multi-agent system consisting of $N$ agents. The dynamics of each agent is

\begin{equation}
	\begin{aligned}
		\begin{cases}
			\mathcal{E}dx_i(t)= & \Big\{\big(A(r_t)+\Delta A(r_t)\big) x_i(t)++B(r_t)[u_i^{f}(t)+f(t, x_i)]\Big\}dt+D_i x_i(t) d\omega(t),\\
			\mathcal{E}dx_0(t)= & \Big\{\big(A(r_t)+\Delta A(r_t)\big) x_0(t)++B(r_t)f(t, x_0)\Big\}dt+D_i x_0(t) d\omega(t),
		\end{cases}
	\end{aligned}\label{1}
\end{equation}
where $x(t)\in\mathbb{R}^{N}$ is the state vector, $u(t)\in\mathbb{R}^{N}$ is the control input , $f(x_i(t),r(t),t) \in  \mathbb{R}^{N\times N} $ is the continuously differentiable vector-valued functions which denote the nonlinear inherent dynamics of  follower $i$ and $\omega(t)$ is a one-dimensional Brownian motion satisfying $ \mathbf{E}\left\{d\omega(t)=0\right\} $ and $ \mathbf{E}\left\{d\omega^{2}(t)=dt\right\} $. $\mathbf{E}\left\{\cdot \right\}$ represents the mathematical expectation of a random process or vector. The matrix $\mathcal{E}\in \mathbb{R}^{N\times N} $ may be singular. It is assumed that $rank(\mathcal{E})\leq N$. Matrices $A(r_t)$, $B(r_t)$ and $D(r_t)$ are known and real with appropriate dimensions where $B(r_t)$ has full column rank.In order to simplify the concept, we denote each potential value of $r_t$ as $r_t = r \in S$, where $A(r_t)$ is represented as $A_r$, $B(r_t)$ as $B_r$, $C(r_t)=C_r$, $D(r_t)=D_r$, $\Delta A(r_t,t)=\Delta A_r(t)$.
In the equation, $\Delta A_r$ represents the structural uncertainty that satisfies $\Delta A_r = M_i F_i(t)N_i$, where $M_i$ and $N_i$ are known matrices, and $F_i(t)$ represents the unknown uncertain matrix function that satisfies $F^{T}_{i}(t)F_i(t) < I$.The fault-tolerant formation control protocol is designed as follows:
\begin{equation}
	u_i^{f}(t)=\eta_i u_i(t)+u_b(t),\quad t>t_f,
\end{equation}
In the equation, $\eta(t)$ represents the actuator efficiency coefficient, defined as $\eta(t) = \text{diag}\left\{\eta_1, \eta_2, \ldots, \eta_n\right\}\in \mathbb{R}^{N\times N}$, where $\eta_i$ denotes the efficiency of the $i$-th actuator. $u_b(t) \in \mathbb{R}^{N\times 1}$ represents the bias fault, and $t_f$ is the time of occurrence of an unknown fault. Additionally, $u_i(t)$ represents the actuator subjected to physical constraints.
Using the properties of the Kronecker product, then we can get
\begin{equation}
\begin{split}\label{004}
		\mathcal{\tilde{E}} dx(t)= \left\{\left(\tilde{A}_r+\Delta \tilde{A}_r\right) x(t)+\tilde{B}_r\left[u^f(t) +f(t, x(t))\right]\right\} dt +\tilde{D}_r x(t)d \omega(t),
\end{split}
\end{equation}
where $$x(t)=\left[x_i(t), x_2(t), \ldots, x_N(t)\right]^T$$, $$f(t, x(t))=\left[f_1(t, x_1(t)), f_2(t, x_2(t)), \ldots, f_N(t, x_N(t))\right]^{T}$$. The function $f(t, x(t))$ is a bounded function and satisfies $|f(t, x(t))| \leq \kappa|x(t)|$, where $\kappa>0$. $\mathcal{\tilde{E}}=I_N \otimes \mathcal{E}$, $\tilde{A}_r=I_N \otimes A_r$, $\Delta \tilde{A}_r=I_N \otimes \Delta A_r$, $ \tilde{B}_r=I_N \otimes B_r$,$\tilde{D}_r=I_N \otimes D_r$. Define the tracking error as $y_i(t)=x_i(t)-x_0(t)-h_i,i=1,12,\ldots,N $ , $u(t)=\left[u_1(t),u_2(t),\ldots,u_N(t)\right]^{T}$, $y(t)=[y_1(t),y_2(t),\ldots,y_N(t)]^{T}$.
\begin{equation}
	\begin{aligned}
		\mathcal{\tilde{E}} dy(t)= \left\{\left(\tilde{A}_r+\Delta \tilde{A}_r\right) y(t)+\tilde{B}_r\left[u^f(t) +f(t, y(t))\right]\right\} dt +\tilde{D}_r y(t)d \omega(t),
	\end{aligned}\label{007}
\end{equation}
\begin{Assumption}
	Assuming in the topology described by $G$, the leader is reachable by all followers, meaning that for each follower, there exists at least one directed path from the leader to it.
\end{Assumption}
\begin{Definition}\cite{zhang2017sliding}
	Consider the unforced nominal singular MJSs as
	\begin{equation}
		\mathcal{E} d x(t)=\mathcal{A}_i x(t) d t+\mathcal{D}_i x(t) d \omega(t) .\label{0007}
	\end{equation}
	System ($\ref{0007}$) is said to be: (i) regular if $\operatorname{det}\left(s \mathcal{E}-\mathcal{A}_i\right)$ is not identically zero; (ii) impulse-free if $\operatorname{deg}\left\{\operatorname{det}\left(s \mathcal{E}-\mathcal{A}_i\right)\right\}=$ $\operatorname{Rank}(\mathcal{E})$; (iii) stochastically stable if for any initial condition $x_0$ and $r_0 \in N, E\left\{\int_0^{\infty}\|x(t)\|^2 d t \mid x_0, r_0\right\}<\infty$ holds; (iv) stochastically admissible, if it is regular, impulse-free, and stochastically stable.\label{def1}
\end{Definition}
\begin{Definition}
	The time-varying formation of agent $i$ is specified by the vector $h_i(t) = [h_{ip}(t), h_{iv}(t)]^T \in \mathbb{R}^{2\times1}$, where $h_i(t)$ is continuously differentiable, and $h_{ip}(t)$ and $h_{iv}(t)$ are the formation position and velocity components, respectively, with $\dot{h}{ip}(t) = h{iv}(t)$. For any initial conditions and $i = 1,2,\ldots, N$, there exists a constant $\theta \ge 0$ such that
	\begin{equation}
		\lim_{t\to \infty}\| x_{i}-h_{i}-x_{0} \| \leq \theta,
	\end{equation}
	where $\theta$ represents the tracking error bound for the multi-agent system formation.
\end{Definition}
\begin{lemma}
	For any vectors $X$ and $Y$, $Q> 0$ belongs to $\mathbb{R}^{n\times n}$. The following inequality holds
	\begin{equation}
		2X^TY\leq X^TQX+Y^TQ^{-1}Y.
	\end{equation}
\end{lemma}
\begin{lemma}\cite{Hirt1974}
	Let $M, F,N$ and $P$ be real matrices of appropriate dimensions with $P>0,F^{\mathrm{T}}F\le I$ and a scalar $\varepsilon>0.$ Then
	\begin{equation}
		MFN+N^{\mathrm{T}}F^{\mathrm{T}}M^{\mathrm{T}}\leq\varepsilon MP^{-1}M^{\mathrm{T}}+\frac{1}{\varepsilon}N^{\mathrm{T}}F^{\mathrm{T}}PFN.
	\end{equation}
\end{lemma}
\begin{remark}
	From Lemma 2,when $F\:=\:I,$ it follows that MN +
	$N^{\mathrm{T}}M^{\mathrm{T}}\:\leq\:\varepsilon MP^{-1}M^{\mathrm{T}}\:+\:\frac{1}{\varepsilon}N^{\mathrm{T}}PN,$and when $P\:=\:I,$ it follows that
	\begin{equation}
		MFN+N^{\mathrm{T}}F^{\mathrm{T}}M^{\mathrm{T}}\leq\varepsilon MM^{\mathrm{T}}+\frac{1}{\varepsilon}N^{\mathrm{T}}N.
	\end{equation}
\end{remark}
\section{MAIN RESULTS}\label{sec3}
This section investigates admissible formation analysis of a nonlinear singular stochastic  Markovian MASs \eqref{004}. Firstly, we devise an integral-type switching surface function. Subsequently, we derive the sufficient conditions for admissible formation of continuous-time sliding mode dynamics \eqref{edx}, followed by the design of a sliding mode fault-tolerant control law. Finally, building upon the continuous-time sliding mode dynamics \eqref{edx}, an event-triggered sliding mode dynamics \eqref{edy3} is established, and the sufficient conditions for admissible formation are derived.
\subsection{Switching surface design}
In order to achieve consistent tracking formation control, the distributed integral sliding mode surface is designed as follows:
\begin{equation}
\begin{split}
	s_{i}(t)\triangleq&\mathcal{B}_i^T\hat{\mathcal{P}}_i\mathcal{E}y_{i}(t)-\int_0^t\mathcal{B}_i^T\hat{\mathcal{P}}_i\mathcal{A}_iy_{i}(s)ds-\int_0^t\mathcal{B}_i^T\hat{\mathcal{P}}_i\mathcal{B}_i\mathcal{K}_i(\sum_{j\in N_{i}}a_{i j}(x_{j}(s)-x_{i}(s))\\
&+b_{i0}(x_{i}(s)-x_{0}(s)))ds,
\end{split}
\end{equation}
where $\hat{P}_{i}$ and $K_i$ are the real matrices to be designed with $B^{T}\hat{P}_{i}B_{i}$ satisfying nonsingularity. define $s=\left[s_1,s_2,\ldots,s_n\right] $. It should be pointed out that due to $B_i$ with full column rank, the nonsingularity of $B^{T}\hat{P}_{i}B_{i}$ can be ensured by $\hat{P}_{i}> 0 $ , $\forall i \in S$.
\begin{equation}
	s(t)\triangleq \mathcal{\tilde{B}}_i^T\hat{\tilde{\mathcal{P}}}_i\mathcal{\tilde{E}}y(t)-\int_0^t\mathcal{\tilde{B}}_i^T\hat{\tilde{\mathcal{P}}}_i\mathcal{\tilde{A}}_i y(s)ds-\int_0^t\mathcal{\tilde{B}}_i^T\hat{\tilde{\mathcal{P}}}_i\mathcal{\tilde{B}}_i\mathcal{\tilde{K}}_i\otimes \mathcal{H} y(s) ds,\label{00012}
\end{equation}
Considering the solution of system ($\ref{007}$) yields
\begin{equation}
	\begin{aligned}
		{\mathcal{\tilde{E}}}y(t)= {\mathcal{\tilde{E}}}y(0)+\int_{0}^{t}[({\mathcal{\tilde{A}}}_{i}+\Delta{\mathcal{\tilde{A}}}_{i}(t))y(t)+{\mathcal{\tilde{B}}}_{i}(u^{f}(t)+f_{i})]d t+\int_0^t\mathcal{\tilde{D}}_i y(t)d\omega(t),
	\end{aligned}
\end{equation}
where $\int_0^t \mathcal{D}_i y(t) d \omega(t)$ stands for the Itô's stochastic integral. Here, if $\mathcal{B}_i^T \hat{\mathcal{P}}_i \mathcal{D}_i=0$, then it follows from (4) and (5) that
\begin{equation}
	\begin{aligned}
		s(t)=& \mathcal{\tilde{B}}_i^T\hat{\tilde{\mathcal{P}}}_i\mathcal{\tilde{E}} y(0)+\int_0^t \mathcal{\tilde{B}}_i^T\hat{\tilde{\mathcal{P}}}_i\left[\left(\Delta \tilde{A}_i(t)-\mathcal{\tilde{B}}_i\mathcal{\tilde{K}}_i\otimes\mathcal{H}\right) y(t)+\mathcal{\tilde{B}}_i\left(\eta u(t)+u_b(t)\right.\right.\\&\left. \left.+f_i\right)\right] d t\label{St}
	\end{aligned}
\end{equation}
When the integral sliding mode surface take place, then we have $s(t) = 0$ and $\dot{s}(t) = 0$. Thus, the equivalent controller can be represented as
\begin{equation}
	u_{e q}(t)= \eta^{-1}\left[\mathcal{\tilde{K}}_i\otimes\mathcal{H} y(t)-\left(\mathcal{\tilde{B}}_i^T\hat{\tilde{\mathcal{P}}}_i\mathcal{\tilde{B}}_i\right)^{-1} \mathcal{\tilde{B}}_i^T\hat{\tilde{\mathcal{P}}}_i \Delta \tilde{A}_{i} y(t)-f_{i}-u_b(t)) \right]\label{ueq}
\end{equation}
With the equivalent controler ($\ref{ueq}$),the closed-loop sliding mode dynamics can be obtained as
\begin{equation}
	\begin{aligned}
		{\mathcal{\tilde{E}}}d y(t)=&\left[\mathcal{\tilde{A}}_{i}+\mathcal{\tilde{B}}_i\mathcal{\tilde{K}}_i\otimes\mathcal{H}+\Delta\mathcal{\tilde{A}}_{i}(t)-\mathcal{\tilde{B}}_i\left(\mathcal{\tilde{B}}_i^T\hat{\tilde{\mathcal{P}}}_i\mathcal{\tilde{B}}_i\right)^{-1} \mathcal{\tilde{B}}_i^T\hat{\tilde{\mathcal{P}}}_i \Delta \tilde{A}_{i}\right] y(t)d t\\ &+{\mathcal{\tilde{D}}}_{i}y(t)d\omega(t), \label{edx}
	\end{aligned}
\end{equation}
Therefore, the stochastic stability of systems (\ref{edx})  will be investigated in the following part
\subsection{Sliding mode analysis}
The aim of this section is to propose a non-singular integral sliding mode surface based on a non-singular integral sliding mode surface that allows the admissible formation problem to be satisfied in the presence of Markovian switching and random disturbances.
\begin{theorem}
	If there exist matrices $\tilde{\mathcal{P}}_i>0, Q_r > 0, R_1 > 0, P_r>0 $, symmetric matrices
	$\hat{\tilde{\mathcal{P}}}_{i}>0$, nonsingular matrices $\tilde{\mathcal{Q}_{i}}$ $\hat{\mathcal{Q}}_{i}$, and positive scalars $\varepsilon_{1i}$,
	$\varepsilon_{2i}$, $\forall i\in{\mathcal{S}}$, satisfying  
	\begin{equation}
		\begin{aligned}
			&\Phi_{r11}=\begin{bmatrix}
				\Sigma_{1i}&\Sigma_{2i}\\
				*&\Sigma_{3i}
			\end{bmatrix}<0,
			\begin{bmatrix}
				-\hat{\mathcal{P}}_i^{-1}&\mathcal{M}_i\\
				*&-\varepsilon_{2i}\mathcal{I}
			\end{bmatrix}<0,\\
			\Phi_{r} =&\left[\begin{array}{ccccc}
				\Phi_{ r11} & H{\mathcal{\tilde{E}}} &0&0&0  \\
				* & (N-H){\mathcal{\tilde{E}}} &0 &0&0\\
				* & * & \tilde{Q}_j-N{\mathcal{\tilde{E}}}&0&0\\
				* & * & * & \frac{\tilde{R}}{d} & *\\
				* & * & * & * &d {\mathcal{\tilde{E}}}^{T}\tilde{R}{\mathcal{\tilde{E}}}
			\end{array}\right]<0,
		\end{aligned}
	\end{equation}
	where
	\begin{equation}
		\begin{aligned}
			\Sigma_{ 1i}=&\sum_{j=1}^s \pi_{i j}(b) \mathcal{\tilde{E}}^T \bar{P}_j+\bar{P}_i\left(\tilde{\mathcal{A}}_i +\mathcal{\tilde{B}}_i\mathcal{\tilde{K}}_i\otimes\mathcal{H}\right) +\left(\tilde{\mathcal{A}}_i +\mathcal{\tilde{B}}_i\mathcal{\tilde{K}}_i\otimes\mathcal{H}\right)^T  \bar{P}_i\\&+{\mathcal{\tilde{D}}}^{T}_{i} \tilde{\mathcal{P}}_i {\mathcal{\tilde{D}}}_{i}+\tilde{Q}_j,\\
			\Sigma_{ 2i}&=\begin{bmatrix}\mathcal{B}_i&&\mathcal{P}_i^T\mathcal{M}_i&&\varepsilon_{1i}\mathcal{N}_i^T&&\mathcal{N}_{i}^{T}\mathcal{F}_{i}^{T}(t)\mathcal{M}_{i}^{T}\end{bmatrix},\\
			\Sigma_{ 3i}&=-diag\left\{\mathcal{B}_{i}^{T}\hat{\cal P}_{i}\mathcal{B}_{i},\,\varepsilon_{1i}\mathcal{I},\,\varepsilon_{1i}\mathcal{I},\,\varepsilon_{2i}^{2}\mathcal{I}\right\},
		\end{aligned}
	\end{equation}
	the sliding mode dynamics ($\ref{edx}$) is stochastically admissible.
\end{theorem}
\begin{proof}
	Choose the following Markovian switched Lyapunov functional candidate as
	\begin{equation}
		V(x(t),r(t),t)=V_1+V_2+V_3
	\end{equation}
	where
	\begin{equation}
		\begin{aligned}
			& V_1=y^T(t){\mathcal{\tilde{E}}}^{T} \left(I \otimes P_r\right)  {\mathcal{\tilde{E}}}y(t), \\
			& V_2=\int_{t-d}^t y^T(s)\left(I \otimes Q_r\right) y(s) \mathrm{d} s \\
			& V_3=\int_{t-d}^t \int_{t+\theta}^t \dot{y}^T(s)\left(I \otimes R_r\right) \dot{y}(s) \mathrm{d} s \mathrm{~d} \theta,
		\end{aligned}
	\end{equation}
	Let $\mathcal{D}$ refer to the weak infinitesimal operator of the SP $\{(y(t), r(t)), t \geq 0\}$. Then, at time $t$, the infinitesimal operator generating from the point $(y(t), r(t)=i, t)$ is characterized by
	\begin{equation}
\begin{split}
		\mathbf{E}[\mathcal{L}V(y(t), r(t), t)] 
		= &\lim _{\Delta \rightarrow 0} \frac{1}{\Delta}[\mathcal{E}\{V(y(t+\Delta), r(t+\Delta), t+\Delta) \mid y(t), r(t)=i\} \\&
		-V(y(t), r(t), t)], \quad i \in N
\end{split}
	\end{equation}
	For each $r(t)=i \in N$, adopting the law of total probability and conditional expectation, provides,
	\begin{equation}
		\begin{split}
			\mathbf{E}&[\mathcal{L}V_1(y(t), r(t), t)] = \lim_{\Delta \rightarrow 0} \frac{1}{\Delta} \Bigg[\mathbf{E} \Bigg\{\sum_{j=1,j \neq i}^{N} \mathrm{Pr}\Big\{r_{n+1}=j, h_{n+1} \leq h\\&+\Delta \mid r_n=r, h_{n+1}>h\Big\}
 y^{\top}(t+\Delta) {\mathcal{\tilde{E}}}^{T} P_j {\mathcal{\tilde{E}}} y(t+\Delta) \\
			&+ \mathrm{Pr}\Big\{r_{n+1}=r, h_{n+1}>h+\Delta \mid r_n=r, h_{n+1}>h\Big\} y^{T}(t+\Delta) {\mathcal{\tilde{E}}}^{T} P_r {\mathcal{\tilde{E}}}y(t+\Delta)\Bigg\}\\& - y^{T}(t) {\mathcal{\tilde{E}}}^{T} P_r {\mathcal{\tilde{E}}} y(t)\Bigg]\\
			=&  \lim _{\Delta \rightarrow 0} \frac{1}{\Delta}\Bigg[\mathbf{E} \Bigg\{\sum_{j=1, j \neq i}^N \frac{\operatorname{Pr}\left\{r_{n+1}=j, r_n=r\right\}}{\operatorname{Pr}\left\{r_n=r\right\}}\frac{\operatorname{Pr}\left\{h<h_{n+1} \leq h+\Delta \mid r_{n+1}=j, r_n=r\right\}}{\operatorname{Pr}\left\{h_{n+1}>h \mid r_n=r\right\}}
\\&\times y^{\top}(t+\Delta) {\mathcal{\tilde{E}}}^{T} P_j {\mathcal{\tilde{E}}}y(t+\Delta)+\frac{\operatorname{Pr}\left\{h_{n+1}>h+\Delta \mid r_n=r\right\}}{\operatorname{Pr}\left\{h_{n+1}>h \mid r_n=r\right\}}  y^{\top}(t+\Delta) {\mathcal{\tilde{E}}}^{T} P_r {\mathcal{\tilde{E}}} y(t+\Delta)\Bigg\}\\&-y^{\top}(t){\mathcal{\tilde{E}}}^{T} P_r {\mathcal{\tilde{E}}} y(t)\Bigg] \\
			=&  \lim _{\Delta \rightarrow 0} \frac{1}{\Delta}\Bigg[\mathbf{E} \Bigg\{\sum_{j=1, j \neq i}^N \frac{q_{i j}\left(G_i(h+\Delta)-G_i(h)\right)}{1-G_i(h)} y^{\top}(t+\Delta) {\mathcal{\tilde{E}}}^{T} P_j  {\mathcal{\tilde{E}}}y(t+\Delta)\\
			&+\frac{1-G_i(h+\Delta)}{1-G_i(h)} y^{\top}(t+\Delta) {\mathcal{\tilde{E}}}^{T} P_r {\mathcal{\tilde{E}}} y(t+\delta)\Bigg\}-y^{\top}(t) {\mathcal{\tilde{E}}}^{T} P_r {\mathcal{\tilde{E}}}y(t)\Bigg],
		\end{split}
	\end{equation}
	where $G_i(h)$ represents the cumulative distribution function(CDF) of the sojourn-time $h$ when the system stays in mode $i$, and $$q_{i j}:=\frac{\operatorname{Pr}\left\{r_{n+1}=j, r_n=i\right\}}{\operatorname{Pr}\left\{\left[r_n=i\right\}\right.}=\operatorname{Pr}\left\{r_{n+1}=j \mid r_n=\right.i)$$  refers to the probability intensity of the system switching from mode $i$ to mode $j$. With a small $\Delta$, the first-order approximation of $y(t+\Delta)$ is
	\begin{equation}
		{\mathcal{\tilde{E}}}y(t+\Delta)={\mathcal{\tilde{E}}}y(t)+\bar{\mathcal{A}}_i(t)y(t)\Delta+{\mathcal{\tilde{D}}}_iy(t)\Delta\omega(t), 
	\end{equation}
	Based upon Eq.(22) and Eq.(23), we have
	\begin{equation}
		\begin{aligned}
			\mathcal{L}V_1&(y(t), r(t), t) =\lim _{\Delta \rightarrow 0} \frac{1}{\Delta} \mathbf{E} \Bigg\{ \sum_{j=1, j \neq i}^N \frac{q_{i j}\left(G_i(h+\Delta)-G_i(h)\right)}{1-G_i(h)}\\&\times\bigg[y^{T}(t){\mathcal{\tilde{E}}}^{T}\tilde{\mathcal{P}}_{j}{\mathcal{\tilde{E}}}y(t)+2y^{T}(t)\bar{\mathcal{A}}_{i}^{T}(t)\tilde{\mathcal{P}}_{j}{\mathcal{\tilde{E}}}y(t)\Delta \\
			&+2y^{T}(t){\mathcal{\tilde{D}}}_{i}^{T}\tilde{\mathcal{P}}_{j}{\mathcal{\tilde{E}}}y(t)\Delta\omega(t)+2y^{T}(t)\bar{\mathcal{A}}_{i}^{T}(t)\tilde{\cal P}_{j}{\mathcal{\tilde{D}}}_{i}y(t)\Delta\omega(t)\Delta\\&+y^{T}(t)\bar{\mathcal{A}}_{i}^{T}(t)\tilde{\cal P}_{j}\bar{\mathcal{A}}_{i}(t)y(t)\Delta^{2} 
			+y^{T}(t){\mathcal{\tilde{D}}}_{i}^{T}\tilde{\mathcal{P}}_{j}{\mathcal{\tilde{D}}}_{i}y(t)\Delta^{2}\omega(t)\bigg]\\&+\frac{G_{i}(h)-G_{i}(h+\Delta)}{1-G_{i}(h)}y^{T}(t){\mathcal{\tilde{E}}}^{T}\tilde{\mathcal{P}}_{i}{\mathcal{\tilde{E}}}y(t)\\
			&+\frac{1-G_{i}(h+\Delta)}{1-G_{i}(h)}\bigg[2y^{T}(t)\bar{\mathcal{A}}_{i}^{T}(t)\tilde{\mathcal{P}}_{i}{\mathcal{\tilde{E}}}y(t)\Delta+2y^{T}(t){\mathcal{\tilde{D}}}_{i}^{T}\tilde{\mathcal{P}}_{i}{\mathcal{\tilde{E}}}y(t)\Delta\omega(t) \\
			&+2y^{T}(t)\bar{\mathcal{A}}_{i}^{T}(t)\tilde{\cal P}_{i}{\mathcal{\tilde{D}}}_{i}y(t)\Delta\omega(t)\Delta+y^{T}(t)\bar{\mathcal{A}}_{i}^{T}(t)\tilde{\cal P}_{i}\bar{\mathcal{A}}_{i}(t)y(t)\Delta^{2}\\&+y^{T}(t){\mathcal{\tilde{D}}}_{i}^{T}\tilde{\cal P}_{i}{\mathcal{\tilde{D}}}_{i}y(t)\Delta^{2}\omega(t)\bigg]\Bigg\} .	
		\end{aligned}
	\end{equation}
	From the property of CDF, we have
	\begin{equation}
		\begin{aligned}
			&\lim_{\Delta\to0}\frac{1-G_i(h+\Delta)}{1-G_i(h)}=1, \lim_{\Delta\to 0}\frac{G_i(h)-G_i(h+\Delta)}{1-G_i(h)}=0,\\& \lim_{\Delta\to 0}\frac{G_i(h+\Delta)-G_i(h)}{\Delta(1-G_i(h))}=\lambda_i(h).
		\end{aligned}
	\end{equation}
	With $\pi_{ij}(h)=q_{ij} \lambda_{j}(h),i\neq j$, we can obtain 
	\begin{equation}
		\begin{aligned}
			&\lim_{\Delta\rightarrow0}\frac{1}{\Delta}\sum_{j=1,j\neq i}^{\mathcal{N}}\frac{q_{ij}(G_{i}(h+\Delta)-G_{i}(h))}{1-G_{i}(h)}[y^{T}(t){\mathcal{\tilde{E}}}^{T}\tilde{\cal P}_{j}{\mathcal{\tilde{E}}}y(t)]\\
			&=\sum_{j=1,j\neq i}^{N}q_{ij}\lambda_{i}(h)[y^{T}(t){\mathcal{\tilde{E}}}^{T}\tilde{\mathcal{P}}_{j}\mathcal{E}y(t)] \\
			&=\sum_{i-1~i\neq i}^{N}\pi_{ij}(h)y^{T}(t){\mathcal{\tilde{E}}}^{T}\tilde{\cal P}_{j}{\mathcal{\tilde{E}}}y(t).
		\end{aligned}
	\end{equation}
	By $It\hat{o}'s$ formula, it is noted that
	\begin{equation}
		\begin{aligned}
			&\mathbf{E}\left\{\lim_{\Delta\rightarrow 0}\frac{1}{\Delta}\sum_{j=1,j\neq i}^{N}\frac{q_{ij}(G_{i}(h+\Delta)-G_{i}(h))}{1-G_{i}(h)}\left[2y^{T}(t){\mathcal{\tilde{D}}}_{i}^{T}\tilde{\mathcal{P}}_{j} \mathcal{E}y(t)\Delta\omega(t)\right]\right\}=0, \\
			&\mathbf{E}\left\{\lim_{\Delta\rightarrow 0}\frac{1}{\Delta}\sum_{j=1,j\neq i}^{N}\frac{q_{ij}(G_{i}(h+\Delta)-G_{i}(h))}{1-G_{i}(h)}\left[y^{T}(t){\mathcal{\tilde{D}}}_{i}^{T}\tilde{\mathcal{P}}_{j}{\mathcal{\tilde{D}}}_{i} y(t)\Delta^{2}\omega(t)\right]\right\}\\
			&=\mathbf{E}\left\{\lim_{\Delta\to0}\frac{1}{\Delta}\sum_{j=1,j\neq i}^{\mathcal{N}}\frac{q_{ij}(G_{i}(h+\Delta)-G_{i}(h))}{1-G_{i}(h)}\left[y^{T}(t){\mathcal{\tilde{D}}}_{i}^{T} \tilde{\cal P}_{j}{\mathcal{\tilde{D}}}_{i}y(t)\Delta\right]\right\}=0.
		\end{aligned}
	\end{equation}
	Thus, it follows from (24)-(27) that
	\begin{equation}
		\begin{aligned}
			\mathcal{L}V_1(y(t), r(t), t)=&y^T(t) \sum_{j=1}^s \pi_{i j}(b) \mathcal{\tilde{E}}^T \bar{P}_j y(t) +y^T(t) \bar{P}_i \bar{\mathcal{A}}_i +\bar{\mathcal{A}}_i^T  \bar{P}_i y(t)\\&+y(t)^T{\mathcal{\tilde{D}}}^{T}_{i} \tilde{\mathcal{P}}_i {\mathcal{\tilde{D}}}_{i}y(t) .
		\end{aligned}
	\end{equation}
	with 
	\begin{equation}
		\begin{aligned}
			\bar{\mathcal{A}}_i=&\mathcal{\tilde{A}}_{i}y(t)+\mathcal{\tilde{B}}_i\mathcal{\tilde{K}}_i\otimes\mathcal{H}y(t)+\Delta\mathcal{\tilde{A}}_{i}(t)y(t)-\mathcal{\tilde{B}}_i\left(\mathcal{\tilde{B}}_i^T\hat{\tilde{\mathcal{P}}}_i\mathcal{\tilde{B}}_i\right)^{-1} \mathcal{\tilde{B}}_i^T\hat{\tilde{\mathcal{P}}}_i \Delta \tilde{A}_{i}y(t),\\
			\bar{\mathcal{P}}_{i}=&\tilde{\cal P}_{i}{\mathcal{\tilde{E}}}+{\cal U}^{T}\tilde{\cal Q}_{i}{\cal V}^{T},\bar{\mathcal{P}}_{i}^{-1}=\hat{\cal P}_{i}{\mathcal{\tilde{E}}}^{T}+{\cal V}\hat{\cal Q}_{i}{\cal U},
		\end{aligned}
	\end{equation}
	On the other hand, for $V_2$ and $V_3$, we have
	\begin{equation}
		\begin{aligned}
			\mathcal{L}V_2(y(t), r(t), t)= &  y^T(t) \tilde{Q}_j y(t)+y^T(t-d_m) \tilde{Q}_j y(t-d_m) - y^T(t-h) \tilde{U}_i y(t-h)
		\end{aligned}
	\end{equation}
	and
	\begin{equation}
		\begin{aligned}
			\mathcal{L}V_3(y(t), r(t), t)=d \dot{y}^T(t) {\mathcal{\tilde{E}}}^{T}\tilde{R}{\mathcal{\tilde{E}}}\dot{y}(t)-\int_{t-d}^t \dot{y}^T(s){\mathcal{\tilde{E}}}^{T} \tilde{R}{\mathcal{\tilde{E}}}\dot{y}(s) \mathrm{d} s
		\end{aligned}
	\end{equation}
	According to Eq.(28)-(31),we derive
	\begin{equation}
		\begin{aligned}
			\mathcal{L}V(y(t), r(t), t)=&y^T(t) \sum_{j=1}^s \pi_{i j}(b) \mathcal{\tilde{E}}^T \bar{P}_j y(t) +y^T(t) \bar{P}_i \bar{\mathcal{A}}_i +\bar{\mathcal{A}}_i^T  \bar{P}_i y(t)\\&+y(t)^T{\mathcal{\tilde{D}}}^{T}_{i} \tilde{\mathcal{P}}_i {\mathcal{\tilde{D}}}_{i}y(t) 
			+y^{T}(t)\tilde{Q}_j y(t)-y^{T}(t-d)\tilde{Q}_j y(t-d)\\&+d\dot{y}^T(t) {\mathcal{\tilde{E}}}^{T}\tilde{R}{\mathcal{\tilde{E}}}\dot{y}(t)-\int_{t-d}^t \dot{y}^T(s){\mathcal{\tilde{E}}}^{T} \tilde{R}{\mathcal{\tilde{E}}}\dot{y}(s) \mathrm{d} s. 
		\end{aligned}\label{031}
	\end{equation}
	For $H$, $N$ with compatible dimensions:
	\begin{equation}
		\begin{aligned}
			&2\eta^T(t)N\left[{\mathcal{\tilde{E}}}y(t-d(t))-{\mathcal{\tilde{E}}}y(t-d)-\int_{t-d}^{t-d(t)}{\mathcal{\tilde{E}}}\dot{y}(s)ds\right]=0\\
			&2\eta^T(t)H\left[{\mathcal{\tilde{E}}}y(t)-{\mathcal{\tilde{E}}}y(t-d(t))-\int_{t-d(t)}^t{\mathcal{\tilde{E}}}\dot{y}(s)ds\right]=0\\
		\end{aligned}\label{032}
	\end{equation}
	where
	\begin{equation}
		\begin{aligned}
			N&\triangleq\left[{N_1}^T \,{N_2}^T \,{N_3}^T\,0 \,\,0\right]^T, H\triangleq\left[{H_1}^T\,{H_2}^T\,{H_3}^T\,0 \,\,0\right]^T,\\
			\eta^T(t)&\triangleq\left[y^T(t)\, y^T(t-d(t))\, y^T(t-d)\,\int_{t-d(t)}^t{\dot{y}^T{}\mathcal{\tilde{E}}}^{T}(s)ds\,\,\dot{y}\right].
		\end{aligned}
	\end{equation}
	According to Lemma 1 of , we have
	\begin{equation}
		\begin{aligned}
			2y^{T}(t)\bar{\mathcal{P}}^{T}_{i}\Delta\mathcal{\tilde{A}}_{i}(t)y(t) \leq& \varepsilon_{1i}^{-1}y^{T}(t)\bar{\mathcal{P}}^{T}_{i}{\cal M}_{i}{\cal M}_{i}^{T}\bar{\mathcal{P}}_{i}y(t)\\&+\varepsilon_{1i}y^{T}(t){\cal N}_{i}^{T}{\cal N}_{i}y(t),  \\
			-2y^{T}(t)\bar{\mathcal{P}}^{T}_{i}\mathcal{\tilde{B}}_i\left(\mathcal{\tilde{B}}_i^T\hat{\tilde{\mathcal{P}}}_i\mathcal{\tilde{B}}_i\right)^{-1} \mathcal{\tilde{B}}_i^T\hat{\tilde{\mathcal{P}}}_i \Delta \tilde{A}_{i} \leq& \varepsilon_{2i}^{2}y^{T}(t)\mathcal{\tilde{B}}_i\left(\mathcal{\tilde{B}}_i^T\hat{\tilde{\mathcal{P}}}_i\mathcal{\tilde{B}}_i\right)^{-1} \mathcal{\tilde{B}}_i^Ty(t)  \\&+\frac{1}{\varepsilon_{2i}^{2}}{\cal N}_{i}^{T}{\cal F}_{i}^{T}(t){\cal M}_{i}^{T}\hat{\tilde{\mathcal{P}}}_i{\cal M}_{i}{\cal F}_{i}(t){\cal N}_{i}.
		\end{aligned}\label{033}
	\end{equation}
	From  \eqref{032}–\eqref{033}, Eq.\eqref{031} is simplified as 
	\begin{equation}
		\begin{aligned}
			\mathcal{L}&V(y(t), r(t), t)=y^T(t) \sum_{j=1}^s \pi_{i j}(b) \mathcal{\tilde{E}}^T \bar{P}_j y(t)\\& +y^T(t) \bar{P}_i \bar{\mathcal{A}}_i +\bar{\mathcal{A}}_i^T  \bar{P}_i y(t)\\&+y(t)^T{\mathcal{\tilde{D}}}^{T}_{i} \tilde{\mathcal{P}}_i {\mathcal{\tilde{D}}}_{i}y(t)+y^{T}(t)\tilde{Q}_j y(t) +d\dot{y}^T(t) {\mathcal{\tilde{E}}}^{T}\tilde{R}{\mathcal{\tilde{E}}}\dot{y}(t)\\
			&+2\eta^T(t)N\left[{\mathcal{\tilde{E}}}y(t-d(t))-{\mathcal{\tilde{E}}}y(t-d)-\int_{t-d}^{t-d(t)}{\mathcal{\tilde{E}}}\dot{y}(s)ds\right]\\&+2\eta^T(t)H\left[{\mathcal{\tilde{E}}}y(t)-{\mathcal{\tilde{E}}}y(t-d(t))-\int_{t-d(t)}^t{\mathcal{\tilde{E}}}\dot{y}(s)ds\right]\\
			&-y^{T}(t-d)\tilde{Q}_j y(t-d)-\int_{t-d}^t \dot{y}^T(s){\mathcal{\tilde{E}}}^{T} \tilde{R}{\mathcal{\tilde{E}}}\dot{y}(s)ds\\
			\leq&y^T(t) \sum_{j=1}^s \pi_{i j}(b) \mathcal{\tilde{E}}^T \bar{P}_j y(t) +2y^T(t) \bar{P}_i \left[\mathcal{\tilde{A}}_{i}+\mathcal{\tilde{B}}_i\mathcal{\tilde{K}}_i\otimes\mathcal{H}\right]+y(t)^T{\mathcal{\tilde{D}}}^{T}_{i} \tilde{\mathcal{P}}_i {\mathcal{\tilde{D}}}_{i}y(t)\\&+y^{T}(t)\tilde{Q}_j y(t) 
			+\varepsilon_{1i}^{-1}y^{T}(t)\bar{\mathcal{P}}^{T}_{i}{\cal M}_{i}{\cal M}_{i}^{T}\bar{\mathcal{P}}_{i}y(t)\\&+\varepsilon_{1i}y^{T}(t){\cal N}_{i}^{T}{\cal N}_{i}y(t)+\varepsilon_{2i}^{2}y^{T}(t)\mathcal{\tilde{B}}_i\left(\mathcal{\tilde{B}}_i^T\hat{\tilde{\mathcal{P}}}_i\mathcal{\tilde{B}}_i\right)^{-1} \mathcal{\tilde{B}}_i^Ty(t)  \\
			&+\frac{1}{\varepsilon_{2i}^{2}}{\cal N}_{i}^{T}{\cal F}_{i}^{T}(t){\cal M}_{i}^{T}\hat{\tilde{\mathcal{P}}}_i{\cal M}_{i}{\cal F}_{i}(t){\cal N}_{i}+d\dot{y}^T(t) {\mathcal{\tilde{E}}}^{T}\tilde{R}{\mathcal{\tilde{E}}}\dot{y}(t)\\&-\frac{1}{d}\left[\int_{t-d}^t \dot{y}^T(s){\mathcal{\tilde{E}}}^{T}ds\right]\tilde{R}\left[\int_{t-d}^t{\mathcal{\tilde{E}}}\dot{y}(s)ds\right]\\
			&-y^{T}(t-d)\tilde{Q}_j y(t-d)+2\eta^T(t)N\left[{\mathcal{\tilde{E}}}y(t-d(t))-{\mathcal{\tilde{E}}}y(t-d)-\int_{t-d}^{t-d(t)}{\mathcal{\tilde{E}}}\dot{y}(s)ds\right]\\
			&+2\eta^T(t)H\left[{\mathcal{\tilde{E}}}y(t)-{\mathcal{\tilde{E}}}y(t-d(t))-\int_{t-d(t)}^t{\mathcal{\tilde{E}}}\dot{y}(s)ds\right]\\
			=&\eta^T(t)\Phi_{r}\eta(t)
		\end{aligned}
	\end{equation}
	with $\Phi_{r 11}<0$. Inequation (19) implies $\Phi_i<0$.
	Hence, there exits a scalar $\lambda>0$ such that
	\begin{equation}
		\begin{aligned}
			\mathcal{L}V(y(t), r(t), t) \leq-\lambda\|y(t)\|^2 .
		\end{aligned}
	\end{equation}
	Thus, by Dynkin's formula, we know for any $t>0$
	\begin{equation}
		\mathbf{E} \left\{\int_0^t\|y(s)\|^2 d s\right\} \leq \lambda^{-1} \mathbf{E}\left\{ V(y(0), r(0))\right\} 
	\end{equation}
	Then, from Definition \ref{def1}, the sliding mode dynamics (\ref{edx}) is stochastically stable. This completes the proof.
\end{proof}
\subsection{SMC law design}
The sliding mode control law is designed to study the reachability of the ITSMSs $s(t) = 0$.  The following theorem guarantees that the state responses of dynamic system can be driven onto the predefined sliding switching surfaces in finite time.The distributed formation tracking protocol for follower agent is proposed as
\begin{equation}
	\begin{aligned}
		u(t)= & \tilde{\mathcal{K}}_i \otimes \mathcal{H} y(t)-\rho\left(\left\|\left(\mathcal{\tilde{B}}_i^T\hat{\tilde{\mathcal{P}}}_i\mathcal{\tilde{B}}_i\right)^{-1} \mathcal{\tilde{B}}_i^T\hat{\tilde{\mathcal{P}}}_i \mathcal{M}_i\right\|\left\|\mathcal{N}_i y(t)\right\|\right.\\
		&\left. +\frac{1}{2}\left\|\sum_{j=1, j \neq i}^{\mathcal{N}} \pi_{i j}\left(\mathcal{\tilde{B}}_i^T\hat{\tilde{\mathcal{P}}}_i\mathcal{\tilde{B}}_i\right)^{-1}\right\|\|s(t)\|+\alpha_i\|y(t)\|+\vartheta_i \right) \operatorname{sgn}(s(t))
	\end{aligned}\label{0042}
\end{equation}
with positive constants $\vartheta_i$ and $\tilde{\mathcal{K}}=\rho\mathcal{K} ,\rho=\eta^{-1}$.
\begin{theorem}
	Under protocol (\ref{0042}), if theorem 1 is satisfied and the control parameter $\vartheta>0$, then the multi-agent system (\ref{edx}) can achieve ideal formation tracking with error boundary $\theta = 0$.
\end{theorem}
\begin{proof}
	Based on the sliding function $s(t)$, an appropriate Lyapunov function is chosen such that it acts on the infinitesimal operator $\mathcal{L}$ applied to the function $V(s(t),r(t),t)$.
	\begin{equation}
		V(s(t),r(t),t)=\frac{1}{2}s^{T}(t)\left(\mathcal{\tilde{B}}_i^T\hat{\tilde{\mathcal{P}}}_i\mathcal{\tilde{B}}_i\right)^{-1}s(t)
	\end{equation}
	From (10), the weak infinitesimal operator $\mathcal{L}$ of $V(s(t), r(t), i)$ is given by
	\begin{equation}
		\begin{aligned}
			\mathcal{L}V(s(t), i)=&\lim _{\Delta \rightarrow 0} \frac{1}{\Delta}\left[\frac{1}{2} \sum_{j=1, j \neq i}^s \mathrm{Pr}\left\{r(t+\Delta)=j \mid r(t)=i\right\}\right.\\ &\left. s^T(t+\Delta) \left(\mathcal{\tilde{B}}_i^T\hat{\tilde{\mathcal{P}}}_i\mathcal{\tilde{B}}_i\right)^{-1}s(t+\Delta)+\frac{1}{2} \mathrm{Pr}\left\{r(t+\Delta)=i \mid r(t)=i\right\}\right.\\
			& \left.\times s^T(t+\Delta)\left(\mathcal{\tilde{B}}_i^T\hat{\tilde{\mathcal{P}}}_i\mathcal{\tilde{B}}_i\right)^{-1} s(t+\Delta)-\frac{1}{2} s^T(t)\left(\mathcal{\tilde{B}}_i^T\hat{\tilde{\mathcal{P}}}_i\mathcal{\tilde{B}}_i\right)^{-1}s(t)\right] \\
			=&\lim _{\Delta \rightarrow 0} \frac{1}{\Delta}\left[\frac{1}{2} \sum_{j=1, j \neq i}^s \frac{q_{i j}\left(G_i(b+\Delta)-G_i(b)\right)}{1-G_i(b)} s^T(t+\Delta) \left(\mathcal{\tilde{B}}_i^T\hat{\tilde{\mathcal{P}}}_i\mathcal{\tilde{B}}_i\right)^{-1}\right.\\ &\left. \times s(t+\Delta)+ \frac{1-G_i(b+\Delta)}{2(1-G_i(b))} s^T(t+\Delta) \right.\\
			& \left.\times\left(\mathcal{\tilde{B}}_i^T\hat{\tilde{\mathcal{P}}}_i\mathcal{\tilde{B}}_i\right)^{-1} s(t+\Delta)-\frac{1}{2} s^T(t)\left(\mathcal{\tilde{B}}_i^T\hat{\tilde{\mathcal{P}}}_i\mathcal{\tilde{B}}_i\right)^{-1} s(t)\right] . \label{uLV1}
		\end{aligned}
	\end{equation}
	From Eq.($\ref{uLV1}$), we have
	\begin{equation}
		\begin{aligned}
			\mathcal{L} V(s(t), i)= & s^T(t)\left(\mathcal{\tilde{B}}_i^T\hat{\tilde{\mathcal{P}}}_i\mathcal{\tilde{B}}_i\right)^{-1} \dot{s}(t) +\frac{1}{2} s^T(t) \sum_{j=1}^s \pi_{r j}(b)\left(\mathcal{\tilde{B}}_i^T\hat{\tilde{\mathcal{P}}}_i\mathcal{\tilde{B}}_i\right)^{-1} s(t) .
		\end{aligned}\label{0045}
	\end{equation}
	On the basis of Eq.($\ref{St}$), we get
	\begin{equation}
		\begin{aligned}
			\dot{S}(t)= & \mathcal{\tilde{B}}_i^T\hat{\tilde{\mathcal{P}}}_i\left[\left(\Delta \tilde{A}_i(t)-\mathcal{\tilde{B}}_i\mathcal{\tilde{K}}_i\otimes\mathcal{H}\right) y(t)+\mathcal{\tilde{B}}_i\left(\eta u(t)+u_b(t)+f_i\right)\right] .
		\end{aligned}\label{0046}
	\end{equation}
	Then, from Eq.($\ref{0045}$) and ($\ref{0046}$) such that
	\begin{equation}
		\begin{aligned}
			\mathcal{L}& V(s(t), r)=  s^T(t)\left(\mathcal{\tilde{B}}_i^T\hat{\tilde{\mathcal{P}}}_i\mathcal{\tilde{B}}_i\right)^{-1} \dot{s}(t)+\frac{1}{2} s^T(t)\sum_{j=1}^s \pi_{r j}(b)\left(\mathcal{\tilde{B}}_i^T\hat{\tilde{\mathcal{P}}}_i\mathcal{\tilde{B}}_i\right)^{-1} s(t) \\
			=&s^{T}(t)\left(\mathcal{\tilde{B}}_i^T\hat{\tilde{\mathcal{P}}}_i\mathcal{\tilde{B}}_i\right)^{-1}\mathcal{\tilde{B}}_i^T\hat{\tilde{\mathcal{P}}}_i\left[\left(\Delta \tilde{A}_i(t)-\mathcal{\tilde{B}}_i\mathcal{\tilde{K}}_i\otimes\mathcal{H}\right) y(t) \right.\\&\left.+\mathcal{B}_{i}\left\{\mathcal{K}_{i}\otimes\mathcal{H}y(t)-\left(\left \| \left(\mathcal{\tilde{B}}_i^T\hat{\tilde{\mathcal{P}}}_i\mathcal{\tilde{B}}_i\right)^{-1}\mathcal{\tilde{B}}_i^T\hat{\tilde{\mathcal{P}}}_i\mathcal{M}_{i} \right \|   \right.\right.\right.\\
			&\left.\left.\left.\times\left \| \mathcal{N}_{i} \right \| \left \| y(t) \right \| +\frac{1}{2}\left \|\sum_{j=1,j\neq i}^{\mathcal{N}}\pi_{ij}\left(\mathcal{\tilde{B}}_i^T\hat{\tilde{\mathcal{P}}}_i\mathcal{\tilde{B}}_i\right)^{-1}\right \| \left \|  s(t)\right \| +\alpha_{i}\left \|  y(t)\right \| +\vartheta_{i}\right)\right.\right.\\&\left.\left. \mathrm{sgn}(s(t))+f_{i}\right\}\right]+\frac{1}{2}s^{T}(t)\sum_{j=1}^{\mathcal{N}}\pi_{ij}(h)\left(\mathcal{\tilde{B}}_i^T\hat{\tilde{\mathcal{P}}}_i\mathcal{\tilde{B}}_i\right)^{-1}s(t)\\
			\leq&-\vartheta_{i}\|s(t)\|.
		\end{aligned}
	\end{equation}
	For each agent $i$, it can reach the sliding surface $S(t) = 0$ in finite time and stay on the sliding surface. It can be seen that the multi-agent system ($\ref{edx}$) achieves ideal formation tracking with an error bound $\theta = 0$. The proof is complete.
\end{proof}
\subsection{Event-triggered sliding mode analysis}
In network communication, an event generator is constructed between a sampler and Zero-order-hold $(\mathrm{ZOH})$ to decide whether to transmit the sampled data $y((k+l) h)$ to the controller through the judgment rule as follows:
\begin{equation}
	\begin{aligned}
		{[y((k+l) h)-y(k h)]^T \Phi_i[y((k+l) h)-y(k h)]}  \leq \sigma(i) y^T((k+l) h) \Phi_i y((k+l) h),
	\end{aligned}
\end{equation}
where $l=1,2, \ldots, \sigma(i) \in[0,1)$ are known triggering parameters and $\Phi_i$ are symmetric positive definite matrices.
On time interval $[t_k, t_{k+1})$ for $k \in N$, define the following sub-intervals: 
\begin{equation}
	\begin{aligned}
		\left\{\begin{array}{l}
			\mathrm{M}_{0, k}=\left[t_k h+\delta_k, t_k h+h+\bar{\delta}\right) \\
			\mathrm{M}_{n, k}=\left[t_k h+n h+\bar{\delta}, t_k h+n h+h+\bar{\delta}\right) \\
			\mathrm{M}_{d, k}=\left[t_k h+d h+\bar{\delta}, t_{k+1} h+\delta_{k+1}\right)
		\end{array}\right.
	\end{aligned}\label{0047}
\end{equation}
where $n=1,2, \ldots, d-1$, and $d$ is constructed such that
\begin{equation}
	t_k h+d h+\bar{\delta}<t_{k+1} h+\delta_{k+1} \leq t_k h+d h+h+\bar{\delta},
\end{equation}
Define $\delta_M \triangleq h+\bar{\delta} $, and
\begin{equation}
	\begin{aligned}
		\delta(t)= \begin{cases}
			t-t_k h, & t \in \mathrm{M}_{0, k}, \\
			t-t_k h-n h & t \in \mathrm{M}_{n, k}, \\
			t-t_k h-d h & t \in \mathrm{M}_{d, k} .
		\end{cases}
	\end{aligned}
\end{equation}
It implies that
\begin{equation}
	\begin{aligned}
		\begin{cases}
			\delta_k \leq \delta(t) \leq h+\bar{\delta}, & t \in \mathrm{M}_{0, k} \\ 
			\delta_k \leq \bar{\delta} \leq \delta(t) \leq h+\bar{\delta}, & t \in \mathrm{M}_{n, k} \\
			\delta_k \leq \bar{\delta} \leq \delta(t) \leq h+\bar{\delta}, & t \in \mathrm{M}_{d, k}
		\end{cases}
	\end{aligned}
\end{equation}
Based on the analysis above and considering the effect of the network-induced transmission delay, here provides the following two case:
\begin{case}
	If $t_k h+h+\bar{\delta} \geq t_{k+1} h+\delta_{k+1}$, denote network-induced time-delay $\delta(t)$ as: $\delta(t)=t-t_k h, t \in\left[t_k h+\delta_k, t_{k+1} h+\delta_{k+1}\right)$, and $e_{k}(t)=0$.
\end{case}
\begin{case}
	If $t_k h+h+\bar{\delta}<t_{k+1} h+\delta_{k+1}$, define
	\begin{equation}
		e_k(t)= \begin{cases}
			0, & t \in \mathrm{M}_0, \\ 
			y\left(t_k h+n h\right)-y\left(t_k h\right), & t \in \mathrm{M}_n, \\
			y\left(t_k h+d h\right)-y\left(t_k h\right), & t \in \mathrm{M}_d
		\end{cases}\label{0051}
	\end{equation}
\end{case}
Assuming that the initial output $y(0)$ is transmitted successfully, it follows that the next release time instant $i_{k+1}$ of the event trigger can be determined by
\begin{equation}
	\begin{aligned}
		t_{k+1}^i=&t_k^i+\min_{l\geq1}\left\{lh| {[y((k+l) h)-y(k h)]^T \Phi_i[y((k+l) h)-y(k h)]}\right. \\&\left. > \sigma(i) y^T((k+l) h) \Phi_i y((k+l) h)\right\}.
	\end{aligned}	
\end{equation}
If the transmission delay exists and $t\in[t_kh+\tau_{t_k},t_{k+1}h+\tau_{t_{k+1}}]$, then the trajectory of the generalized system can be driven to the actual switching surface according to SMC law:
\begin{equation}
	\begin{aligned}
		u(t) = & \mathcal{K}_i \otimes \mathcal{H} y(t_kh) \\
		& - \left\{ \left\| \left(\mathcal{B}_i^T \hat{\mathcal{P}}_i \mathcal{B}_i\right)^{-1} \mathcal{B}_i^T \hat{\mathcal{P}}_i \mathcal{M}_i \right\| \left\| \mathcal{N}_i \right\| \|y(t_kh)\| + \frac{1}{2} \left\| \sum_{j=1, j \neq i}^{\mathcal{N}} \pi_{i j} \left(\mathcal{B}_j^T \hat{\mathcal{P}}_j \mathcal{B}_j\right)^{-1} \right\| \right.
\\&\left. \times\|s(t_kh)\|+ \alpha_i \|y(t_kh)\| + v \right\} \operatorname{sgn}(s(t_kh))
	\end{aligned}
\end{equation}
with $v\ge\|\mathcal{\tilde{B}}_i^T\hat{\tilde{\mathcal{P}}}_i\mathcal{\tilde{B}}_i\mathcal{\tilde{K}}_i\otimes \mathcal{H} \|\|e(t_kh)\|$,the controller is updated.
From $\delta(t)$ and $e_k(t)$, the resultant event-triggered SMC dynamics in \eqref{edx} are deduced as
\begin{equation}
	\begin{aligned}
		{\mathcal{\tilde{E}}}d y(t)=& \left(\mathcal{\tilde{A}}_{i}+\Delta\mathcal{\tilde{A}}_{i}(t)-\mathcal{\tilde{B}}_i\left(\mathcal{\tilde{B}}_i^T\hat{\tilde{\mathcal{P}}}_i\mathcal{\tilde{B}}_i\right)^{-1} \mathcal{\tilde{B}}_i^T\hat{\tilde{\mathcal{P}}}_i \Delta \tilde{A}_{i}\right)y(t)+\mathcal{\tilde{B}}_i\mathcal{\tilde{K}}_i\otimes\mathcal{H}y(t-d)\\&+\mathcal{\tilde{B}}_i\mathcal{\tilde{K}}_i\otimes\mathcal{H}e_{k}(t) d t+{\mathcal{\tilde{D}}}_{i}y(t)d\omega(t), \label{edy2}
	\end{aligned}
\end{equation}
Let $\mathcal{\tilde{G}}_{i}=I-\mathcal{\tilde{B}}_i\left(\mathcal{\tilde{B}}_i^T\hat{\tilde{\mathcal{P}}}_i\mathcal{\tilde{B}}_i\right)^{-1} \mathcal{\tilde{B}}_i^T\hat{\tilde{\mathcal{P}}}_i$, it follows that
\begin{equation}
	\begin{aligned}
		{\mathcal{\tilde{E}}}d y(t)=& \left(\left(\mathcal{\tilde{A}}_{i}+\mathcal{\tilde{G}}_{i}\Delta\mathcal{\tilde{A}}_{i}(t)\right)y(t)+\mathcal{\tilde{B}}_i\mathcal{\tilde{K}}_i\otimes\mathcal{H}y(t-d)+\mathcal{\tilde{B}}_i\mathcal{\tilde{K}}_i\otimes\mathcal{H}e_{k}(t)\right) d t\\&+{\mathcal{\tilde{D}}}_{i}y(t)d\omega(t), \label{edy3}
	\end{aligned}
\end{equation}
\begin{remark}
	Through the previous analysis, system (\ref{edy3}) is converted into a time-delay system due to the network-induced delay. According to triggered schemes \eqref{0047}-\eqref{0051}, for $t \in\left[t_k h+\delta_k, t_{k+1} h+\right.$ $\left.\delta_{k+1}\right)$, we can derive that
	\begin{equation}
		e_k^T(t) \Phi e_k(t) \leq\sigma(i)[y\left(t-d(t)\right)]^T\Phi\left[y\left(t-d(t)\right)\right]
	\end{equation}
\end{remark}
In this section, an integral sliding surface is designed. First, an event-based sliding mode controller is proposed for the closed-loop system. Lyapunov stability theory is used to analyze the stability of the whole system and the accessibility of the sliding mode surface under the event-triggered sliding mode controller.
\begin{theorem}
	If there exist matrices $M_1, M_2, M_3, Q_1>0, Q_2>0, Q>0, R>0, \Phi>0$, nonsingular matrices $P$, positive scalars $\gamma,\sigma_1, d_m, d_M,\varepsilon_i, \forall i \in S$, and
	\begin{equation}
		\begin{aligned}
			E^T P  =P^T E &\geq 0 , {\cal B}_{i}^{T}{\hat{\cal P}}_{i}{\mathcal{\tilde{D}}}_{i}=0,
			\begin{bmatrix}
				-\hat{\cal P}_i&\hat{\cal P}_i\mathcal{M}_i\\
				* & -\varepsilon_{2i}\mathcal{I}\end{bmatrix}&<0,
		\end{aligned}
	\end{equation}
	\begin{equation}
		\begin{aligned}
			\varphi_k  =\left[\begin{array}{ccccc}
				\Xi_{11}+\Gamma+\Gamma^T & M & L_2 & L_1 & \bar{\tau} \ell R \\
				* & -R & 0 & 0 & 0 \\
				* & * & -\gamma^2 I & \bar{H}^T & \bar{\tau} \bar{G}^T R \\
				* & * & * & -I & 0 \\
				* & * & * & * & -\bar{\tau} R
			\end{array}\right]<0,\\
		\end{aligned}
	\end{equation}
	where
	\begin{equation}
		\begin{aligned}
			&L_1=\left[\begin{array}{lllll}
				C & D K & 0 & 0 & D K
			\end{array}\right]^T ,
			L_2=\left[\begin{array}{lllll}
				\bar{G}^T P & 0 & 0 & 0 & 0
			\end{array}\right]^T,\\& \ell=\left[\begin{array}{lllll}
				A & B K & 0 & 0 & B K
			\end{array}\right]^T \\
			&\Gamma=\left[M_1 E ,\left(M_3-M_2\right) E,\left(M_2-M_1\right) E,-M_3 E, 0\right] , \\&M_1=\left[\begin{array}{lllll}
				F_1{ }^T \quad F_2{ }^T\quad F_3{ }^T\quad F_4{ }^T\quad F_5{ }^T
			\end{array}\right]^T, \\
			&M_2=\left[\begin{array}{lll}
				N_1{ }^T \quad N_2{ }^T \quad N_3{ }^T\quad N_4{ }^T\quad N_5{ }^T
			\end{array}\right]^T , \\&M_3=\left[\begin{array}{lllll}
				Z_1{ }^T & Z_2{ }^T & Z_3{ }^T & Z_4{ }^T & Z_5{ }^T
			\end{array}\right]^T ,\\
			&M=\left[\left(\Theta_1^k\right)^T\left(\Theta_2^k\right)^T\left(\Theta_3^k\right)^T\left(\Theta_4^k\right)^T\left(\Theta_5^k\right)^T\right]^T, \Theta_1^1=\sqrt{d_m} F_1+\left(\sqrt{\bar{d}-d_m}\right) N_1,\\
			&  \Theta_1^2=\sqrt{d_m} F_1+\left(\sqrt{\bar{d}-d_m}\right) Z_1 , \Theta_2^1=\sqrt{d_m} F_2+\left(\sqrt{\bar{d}-d_m}\right) N_2,
\\&\Theta_2^2=\sqrt{d_m} F_2+\left(\sqrt{\bar{d}-d_m}\right) Z_2, \Theta_3^1=\sqrt{d_m} F_3+\left(\sqrt{\bar{d}-d_m}\right) N_3,\\& \Theta_3^2=\sqrt{d_m} F_3+\left(\sqrt{\bar{d}-d_m}\right) Z_3 ,\Theta_4^1=\sqrt{d_m} F_4+\left(\sqrt{\bar{d}-d_m}\right) N_4,\\
			&  \Theta_4^2=\sqrt{d_m} F_4+\left(\sqrt{\bar{d}-d_m}\right) Z_4, \Theta_5^1=\sqrt{d_m} F_5+\left(\sqrt{\bar{d}-d_m}\right) N_5,\\&\Theta_5^2=\sqrt{d_m} F_5+\left(\sqrt{\bar{d}-d_m}\right) Z_5 .
		\end{aligned}
	\end{equation}
	with $U$ and $V$ being any matrices with full rank and satisfying $UE = 0$ and $EV = 0$, then the event-trigger sliding mode system ($\ref{edy3}$) is stochastically admissible.  Moreover, the controller gains are computed by $K_i = Y_i X^{-1}$
\end{theorem}
\begin{proof}
	Consider the following form of the Lyapunov function
	\begin{equation}
		V(y(t),r(t),t)=V_1+V_2+V_3
	\end{equation}
	where
	\begin{equation}
		\begin{aligned}
			V_1&=y^T(t){\mathcal{\tilde{E}}}^{T} \left(I \otimes P_r\right)  {\mathcal{\tilde{E}}}y(t),  V_2\\
&=\int_{t-d}^t y^T(s)\left(I \otimes Q_r\right) y(s) \mathrm{d} s , V_3\\&=\int_{t-d}^t \int_{t+\theta}^t \dot{y}^T(s)\left(I \otimes R_r\right) \dot{y}(s) \mathrm{d} s \mathrm{~d} \theta,
		\end{aligned}
	\end{equation}
	with $P_{r}>0$, $Q_{r}>0$,$R_{r}>0$, let $\tilde{P}_{r}=I\otimes P_{r}$, $\tilde{Q}_{r}=I\otimes Q_{r}$, $\tilde{R}_{r}=I\otimes R_{r}$. Then, for $r\in N$ and $t\geq 0$, using the operational properties of stochastic differential equations , one can obtain
	\begin{equation}
		\begin{aligned}
			\mathcal{L}V(y(t), r(t), t)=&y^T(t) \sum_{j=1}^s \pi_{i j}(b) \mathcal{\tilde{E}}^T \bar{P}_j y(t) +2y^T(t)\tilde{\mathcal{P}}_i \left(\left(\mathcal{\tilde{A}}_{i}+\mathcal{\tilde{G}}_{i}\Delta\mathcal{\tilde{A}}_{i}(t)\right)y(t)\right.\\
&\left.+\mathcal{\tilde{B}}_i\mathcal{\tilde{K}}_i\otimes\mathcal{H}y(t-\delta)+\mathcal{\tilde{B}}_i\mathcal{\tilde{K}}_i\otimes\mathcal{H}e_{k}(t)\right) \\
			&+y(t)^T{\mathcal{\tilde{D}}}^{T}_{i} \tilde{\mathcal{P}}_i {\mathcal{\tilde{D}}}_{i}y(t)+y^{T}(t)\tilde{Q}_j y(t)-y^{T}(t-d)\tilde{Q}_j y(t-d)\\&+d\dot{y}^T(t) {\mathcal{\tilde{E}}}^{T}\tilde{R}{\mathcal{\tilde{E}}}\dot{y}(t)-\int_{t-d}^t \dot{y}^T(s){\mathcal{\tilde{E}}}^{T} \tilde{R}{\mathcal{\tilde{E}}}\dot{y}(s) \mathrm{d} s.
		\end{aligned}
	\end{equation}
	Based on the Newton-Leibniz formula, it holds that
	\begin{equation}
		\begin{aligned}
			2\eta^T(t)H\left[\mathcal{\tilde{E}}y(t)-\mathcal{\tilde{E}}y(t-d(t))-\int_{t-d(t)}^t\mathcal{\tilde{E}}\dot{y}(s)ds\right]&=0,\\
			2\eta^T(t)N\left[\mathcal{\tilde{E}}y(t-d(t))-\mathcal{\tilde{E}}y(t-d)-\int_{t-d}^{t-d(t)}\mathcal{\tilde{E}}\dot{y}(s)ds\right]&=0.\\
		\end{aligned}
	\end{equation}
	where $\eta^T(t)=\left[y^T(t)\, y^T(t-d(t))\, y^T(t-d)\,\int_{t-d(t)}^t{\dot{y}^T{}\mathcal{\tilde{E}}}^{T}(s)ds\,\,\dot{y}^{T}\,\,e^T(t)\right]$.
	Based on the event-trigged scheme, it follows that
	\begin{equation}
		\begin{aligned}
			\mathcal{L}V(y(t), r(t), t)=&y^T(t) \sum_{j=1}^s \pi_{i j}(b) \mathcal{\tilde{E}}^T \bar{P}_j y(t) +2y^T(t)\tilde{\mathcal{P}}_i \left(\left(\mathcal{\tilde{A}}_{i}+\mathcal{\tilde{G}}_{i}\Delta\mathcal{\tilde{A}}_{i}(t)\right)y(t)+\right.\\&\left. \mathcal{\tilde{B}}_i\mathcal{\tilde{K}}_i\otimes\mathcal{H}y(t-\delta)+\mathcal{\tilde{B}}_i\mathcal{\tilde{K}}_i\otimes\mathcal{H}e_{k}(t)\right) \\
			&+y(t)^T{\mathcal{\tilde{D}}}^{T}_{i} \tilde{\mathcal{P}}_i {\mathcal{\tilde{D}}}_{i}y(t)+y^{T}(t)\tilde{Q}_j y(t)-y^{T}(t-d)\tilde{Q}_j y(t-d)\\
&+d\dot{y}^T(t) {\mathcal{\tilde{E}}}^{T}\tilde{R}{\mathcal{\tilde{E}}}\dot{y}(t)-\int_{t-d}^t \dot{y}^T(s){\mathcal{\tilde{E}}}^{T} \tilde{R}{\mathcal{\tilde{E}}}\dot{y}(s) \mathrm{d} s
		\end{aligned}
	\end{equation}
	Thus, it is easy to get that there exists a scalar $h > 0 $such that for $i\in N$:
	\begin{equation}
		\mathcal{L}V(y(t), r(t), t)\leq -h \|y(t)\|^{2}, t\in\left[t_k h+\delta_k, t_{k+1} h++\delta_{k+1}\right).
	\end{equation}
	According to Dynkin’s formula, for any $t \ge \delta(t)$
	\begin{equation}
		\mathbf{E}\left\{\int_{\delta(t)}^t\|y(t)\|^2 d t\right\} \leq h^{-1} \mathbf{E} \left\{\mathcal{L}V(y(t), r(t), t)\right\}
	\end{equation}
	Next we will improve the regular and impulse-free of the
	nominal system. Since rank($\mathcal{E}$) $\leq$ n, there exit two nonsingular matrices $\bar{\mathcal{M}}$ and $\bar{\mathcal{N}} \in R^{n×n}$ such that
	\begin{equation}
		\begin{aligned}
			\bar{\mathcal{M}}{\mathcal{EN}} =\begin{bmatrix}\mathcal{I}_r&0\\0&0\end{bmatrix},\bar{\mathcal{M}}\tilde{\mathcal{A}}_i(t)\bar{\mathcal{N}}=\begin{bmatrix}\tilde{\mathcal{A}}_{1i}&\tilde{\mathcal{A}}_{2i}\\\tilde{\mathcal{A}}_{3i}&\tilde{\mathcal{A}}_{4i}\end{bmatrix}, \\
			\bar{\mathcal{M}}^{-T}\mathcal{P}_{i}\bar{\mathcal{N}} =\begin{bmatrix}\mathcal{P}_{1i}&\mathcal{P}_{2i}\\\mathcal{P}_{3i}&\mathcal{P}_{4i}\end{bmatrix},det(\bar{\mathcal{M}})\ne0,det(\bar{\mathcal{N}})\ne 0
		\end{aligned}
	\end{equation}
	define  $\mathcal{P}_{i}=\tilde{\mathcal{P}}_{i}\mathcal{E}+\mathcal{U}^{T}\tilde{\mathcal{Q}}_{i}\mathcal{V}^{T},\mathcal{X}_{i}=\mathcal{P}_{i}^{-1}=\hat{\mathcal{P}}_{i}\mathcal{E}^{T}+\mathcal{V}\hat{\mathcal{Q}}_{i}\mathcal{U},$
	where $\tilde{\mathcal{P}}>0, \tilde{\mathcal{Q}}_i$, and $\hat{\mathcal{Q}}_i$ are nonsingular matrices, $\hat{\mathcal{P}}_i=\hat{\mathcal{P}}_i^T$. In addition, it is obtained that $\mathcal{E}_L^T \tilde{\mathcal{P}}_i \mathcal{E}_L=\left(\mathcal{E}_R^T \hat{\mathcal{P}}_i \mathcal{E}_R\right)^{-1}$ with $\mathcal{E}=\mathcal{E}_L \mathcal{E}_R^T . \mathcal{E}_L \in \mathbb{R}^{n \times r}$ and $\mathcal{E}_R \in \mathbb{R}^{n \times r}$ are of full column rank.
	Since $\operatorname{rank}\{\mathcal{E}\}=\operatorname{rank}\left\{\left[\begin{array}{ll}\mathcal{E} & \mathcal{D}_i\end{array}\right]\right\}$, there exists a matrix $\mathcal{D}_{i d}$, such that $\mathcal{D}_i=\mathcal{E}_L \mathcal{D}_{i d}$. Then, $\mathcal{D}_i^T \tilde{\mathcal{P}}_i \mathcal{D}_i$ can be written as $\mathcal{D}_{i d}^T \mathcal{E}_L^T \tilde{\mathcal{P}}_i \mathcal{E}_L \mathcal{D}_{i d}$,
	due to
	\begin{equation}
		\mathcal{P}_{i}^{T}\tilde{\mathcal{A}}_{i}(t)+\tilde{\mathcal{A}}_{i}^{T}(t)\mathcal{P}_{i}+\sum_{j=1}^{\mathcal{N}}\mu_{ij}(h)\tilde{E}^{T}\mathcal{P}_{j}<0,
	\end{equation}
	it is straightforward to see that $\mathcal{P}_{i2}=0$, Pre- and post-multiplying (65) by $\bar{\mathcal{N}}^{T}$ and $\bar{\mathcal{N}}$ gives
	\begin{equation}
		\begin{bmatrix}
			* & *\\
			* & \tilde{\mathcal{A}}_{4i}^{T}\mathcal{P}_{4i}+\mathcal{P}_{4i}^{T}\tilde{\mathcal{A}}_{4i}\end{bmatrix}<0,
	\end{equation}
	By Lemma 2, it follows from (58) that
	\begin{equation}
		\begin{aligned}
			2x^{T}(t)\mathcal{P}_{i}^{T}\Delta\mathcal{A}_{i}(t)x(t) 
			&\leq\varepsilon_{1i}^{-1}x^{T}(t)\mathcal{P}_{i}^{T}\mathcal{M}_{i}^{T}\mathcal{P}_{i}x(t)+\varepsilon_{1i}x^{T}(t)\mathcal{N}_{i}^{T}\mathcal{N}_{i}x(t),  \\
			-2x^{T}(t){\cal P}_{i}^{T}{\cal B}_{i}({\cal B}_{i}^{T}{\hat{\cal P}}_{i}{\cal B}_{i})^{-1}{\cal B}_{i}^{T}{\hat{\cal P}}_{i}\Delta{\cal A}_{i}(t)&\leq\varepsilon_{2i}^{2}x^{T}(t){\cal P}_{i}^{T}{\cal B}_{i}({\cal B}_{i}^{T}{\cal P}_{i}{\cal B}_{i})^{-1}{\cal B}_{i}^{T}{\cal P}_{i}\\& +\frac{1}{\varepsilon_{2i}^{2}}\mathcal{N}_{i}^{T}\mathcal{F}_{i}^{T}(t)\mathcal{M}_{i}^{T}\hat{\mathcal{P}}_{i}\mathcal{M}_{i}\mathcal{F}_{i}(t)\mathcal{N}_{i}.
		\end{aligned}
	\end{equation}
	Performing a congruence transformation to (11) by $\operatorname{diag}\left\{\mathcal{X}_i, \mathcal{I}, \mathcal{I}, \mathcal{I}, \mathcal{I}\right\}$ yields
	\begin{equation}
		\begin{aligned}
			\left[\begin{array}{cc}
				\Sigma_{1 i} & \Sigma_{2 i} \\
				* & \Sigma_{3 i}
			\end{array}\right]<0,
		\end{aligned}
	\end{equation}
	where
	\begin{equation}
		\begin{aligned}
			\Sigma_{1 i}= & \mathcal{A}_i \mathcal{X}_i+\mathcal{B}_i \mathcal{Y}_i+\mathcal{X}_i^T \mathcal{A}_i^T+\mathcal{Y}_i^T \mathcal{B}_i^T+\mathcal{X}_i^T \mathcal{D}_{i d}^T  \left(\mathcal{E}_R^T \hat{\mathcal{P}}_i \mathcal{E}_R\right)^{-1} \mathcal{D}_{i d} \mathcal{X}_i\\&+\mu_{i i}(h) \mathcal{X}_i^T \mathcal{E}_R\left(\mathcal{E}_R^T \hat{\mathcal{P}}_i \mathcal{E}_R\right)^{-1} \mathcal{E}_R^T \mathcal{X}_i \\
			& +\sum_{j=1, j \neq i}^{\mathcal{N}} \mu_{i j}(h) \mathcal{X}_i^T \mathcal{E}_R\left(\mathcal{E}_R^T \hat{\mathcal{P}}_j \mathcal{E}_R\right)^{-1} \mathcal{E}_R^T \mathcal{X}_i, \\
			\Sigma_{2 i}= & {\left[\mathcal{X}_i^T \mathcal{B}_i, \mathcal{M}_i, \varepsilon_{1 i} \mathcal{X}_i^T \mathcal{N}_i^T, \mathcal{X}_i^T \mathcal{N}_i^T\right] } \\
			\Sigma_{3 i}= & -\operatorname{diag}\left\{\mathcal{B}_i^T \hat{\mathcal{P}}_i \mathcal{B}_i, \varepsilon_{1 i} \mathcal{I}, \varepsilon_{1 i} \mathcal{I}, \varepsilon_{2 i} \mathcal{I}\right\} .
		\end{aligned}
	\end{equation}
	On the other hand,
	\begin{equation}
		\begin{aligned}
			\pi_{i i}(h) \mathcal{X}_i^T \mathcal{E}_R\left(\mathcal{E}_R^T \hat{\mathcal{P}}_i \mathcal{E}_R\right)^{-1} \mathcal{E}_R^T \mathcal{X}_i
			\leq  \pi_{i i}(h)\left(\mathcal{E} \mathcal{X}_i+\mathcal{X}_i^T \mathcal{E}^T-\mathcal{E} \hat{\mathcal{P}}_i \mathcal{E}^T\right)
		\end{aligned}
	\end{equation}
	From the above, it can be concluded that
	\begin{equation}
		\begin{aligned}
			\mathcal{L}&V(y(t), r(t), t)\leq y^T(t) \sum_{j=1}^s \pi_{i j}(b) \mathcal{\tilde{E}}^T \bar{P}_j y(t) +2y^T(t)\tilde{\mathcal{P}}_i \left(\left(\mathcal{\tilde{A}}_{i}+\mathcal{\tilde{G}}_{i}\Delta\mathcal{\tilde{A}}_{i}(t)\right)y(t)+\right.
\\&\left.\mathcal{\tilde{B}}_i\mathcal{\tilde{K}}_i\otimes\mathcal{H}y(t-\delta)+\mathcal{\tilde{B}}_i\mathcal{\tilde{K}}_i\otimes\mathcal{H}e_{k}(t)\right)\\
			&+y(t)^T{\mathcal{\tilde{D}}}^{T}_{i} \tilde{\mathcal{P}}_i {\mathcal{\tilde{D}}}_{i}y(t)+y^{T}(t)\tilde{Q}_j y(t)-e_k^T(t) \Phi e_k(t) +\sigma_1 y(t-d(t))^T\Phi y(t-d(t)) \\
			&+\varepsilon_{1i}^{-1}y^{T}(t)\bar{\mathcal{P}}^{T}_{i}{\cal M}_{i}{\cal M}_{i}^{T}\bar{\mathcal{P}}_{i}y(t)+\varepsilon_{1i}y^{T}(t){\cal N}_{i}^{T}{\cal N}_{i}y(t)+\varepsilon_{2i}^{2}y^{T}(t)\mathcal{\tilde{B}}_i\left(\mathcal{\tilde{B}}_i^T\hat{\tilde{\mathcal{P}}}_i\mathcal{\tilde{B}}_i\right)^{-1} \mathcal{\tilde{B}}_i^Ty(t)  \\
			&+\frac{1}{\varepsilon_{2i}^{2}}{\cal N}_{i}^{T}{\cal F}_{i}^{T}(t){\cal M}_{i}^{T}\hat{\tilde{\mathcal{P}}}_i{\cal M}_{i}{\cal F}_{i}(t){\cal N}_{i}+d\dot{y}^T(t) {\mathcal{\tilde{E}}}^{T}\tilde{R}{\mathcal{\tilde{E}}}\dot{y}(t)\\&-\frac{1}{d}\left[\int_{t-d}^t \dot{y}^T(s){\mathcal{\tilde{E}}}^{T}ds\right]\tilde{R}\left[\int_{t-d}^t{\mathcal{\tilde{E}}}\dot{y}(s)ds\right]\\
			&-y^{T}(t-d)\tilde{Q}_j y(t-d)+2\eta^T(t)N\left[{\mathcal{\tilde{E}}}y(t-d(t))-{\mathcal{\tilde{E}}}y(t-d)-\int_{t-d}^{t-d(t)}{\mathcal{\tilde{E}}}\dot{y}(s)ds\right]\\
			&+2\eta^T(t)H\left[{\mathcal{\tilde{E}}}y(t)-{\mathcal{\tilde{E}}}y(t-d(t))-\int_{t-d(t)}^t{\mathcal{\tilde{E}}}\dot{y}(s)ds\right]\\
			\leq&\eta^{T}(t)\big[\Xi_{11}+\Gamma+\Gamma^{T}+d_{m}M_{1}R^{-1}M_{1}^{T}+(d(t)-d_{m}) M_{2}R^{-1}M_{2}^{T}\\ &+(\bar{d}-d(t))M_{3}R^{-1}M_{3}^{T}+\bar{d}\ell R\ell^{T}\big]\eta(t)
		\end{aligned}
	\end{equation}
	with $\Upsilon=\Xi_{11}+\Gamma+\Gamma^T+d_m M_1R^{-1}M_1^T+(d(t)-d_m)M_2R^{-1}M_2^T+(\bar{d}-d(t)) M_3R^{-1}M_3^T+\bar{d}\ell R\ell^T$.
	\begin{equation}
		\begin{aligned}
			\Xi_{11}=\left[\begin{array}{ccccc}
				P^T A+A^T P+Q_1+Q & P^T B K & 0 & 0 & P^T B K \\
				* & \sigma_1 \Phi-Q & 0 & 0 & 0 \\
				* & * & Q_2-Q_1 & 0 & 0 \\
				* & * & * & -Q_2 & 0 \\
				* & * & * & * & -\Phi+\sigma_2{ }^2 I
			\end{array}\right] \\
		\end{aligned}
	\end{equation}
	By employing linear matrix inequalities (LMIs) in (58), 
	one can obtain $ \Upsilon <0$ , which indicates that the resultant closed-loop system derived from \eqref{edy3} is stochastically stable.   The proof is finished.
\end{proof}
\section{Simulation}\label{sec_5}
To demonstrate the effectiveness and advantages of a proposed event-triggered sliding mode fault-tolerant control strategy for a considered leader-follower multi-agent system, an illustrative example is provided in this section. Consider a multi-agent system consisting of five followers and one leader. Suppose that a 1 or a 0 is used respectively to indicate whether there is information exchange between agents in the communication topology. Let $ F=\{1,2,3,4,5\}$ and $H=\{0\} $ be the set of followers and the set of leaders, respectively, and $S=\{1,2,3\}$ be the communication topology set. To ensure the formation of the multi-agent formation, the preset position of the followers is formed into a closed $R=10$ circular formation, and the leader will be located in the center of the circular formation, and the direction of movement of the leader is the forward direction of the formation. The final states of the follower and leader (position state and formation state) are represented by different colored lines. The position of each follower relative to the leader is as follows:
\begin{equation}
	\begin{aligned}
		h_{i}=R\left[cos(\frac{2\pi}{5}i),sin(\frac{2\pi}{5}i)\right]^{T}, i=1,2,\ldots,5,
	\end{aligned}
\end{equation}
Set the initial location of the leader and follower agents to: $x_{0}(0)=[0,0]$, $x_{1}(0)=[-6,13]$, $x_{2}(0)=[-14,7]$, $x_{3}(0)=[-13,-13]$, $x_{4}(0)=[8,-11]$, $x_{5}(0)=[5,4]$; The initial speed of all agents is set to 0. It is assumed that a partial actuator failure situation may occur in a multi-agent system, and the unknown time-varying failure rate in the five follower nodes satisfies the following condition $0<\eta_i<1$, which means that all five followers encounter time-varying actuator failures.
In the simulation, the control input of the leader is set as $u_{0}=2sin(10t)$ and the nonlinear function $f=e^{t}\sin(10t)$. According to the event triggering condition taken into account in the communication process of the multi-agent formation system above, the event triggering threshold is $\sigma=0.1$ in this experiment, and $\alpha=0.3$ and $\varepsilon=0.2$ are selected as the sliding mode control law.

\begin{figure}[htbp]
	\centering
	\includegraphics[width=0.9\linewidth]{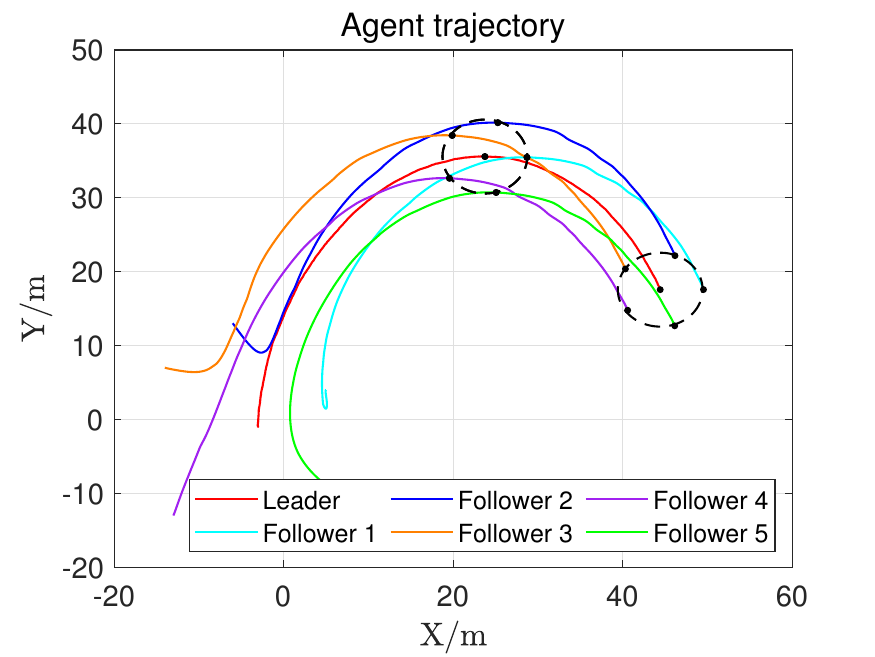}
	\caption{Trajectory of the leader agent and follower agents}
	\label{fig2}
\end{figure}
\begin{figure}[htbp]
	\centering
	\includegraphics[width=0.9\linewidth]{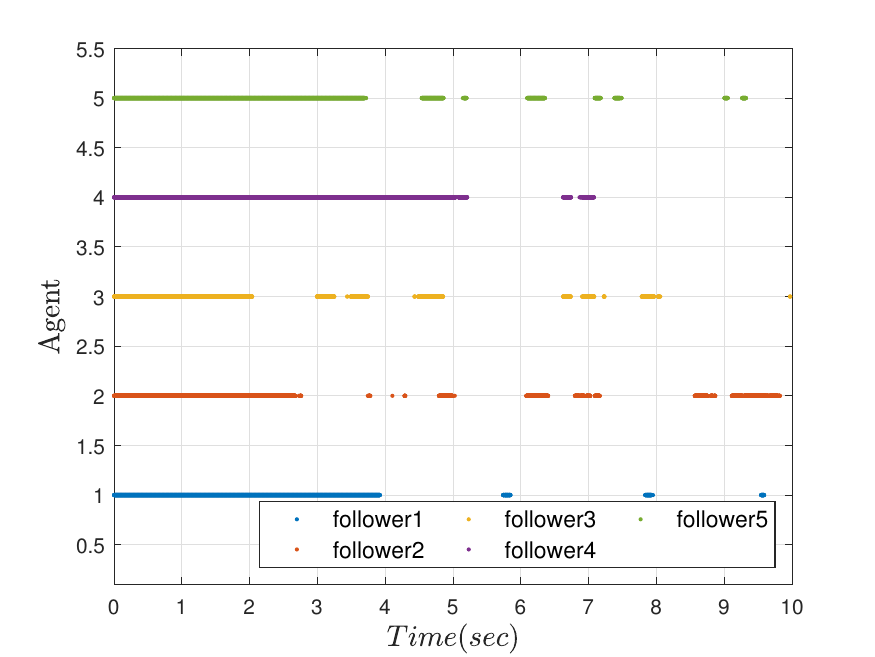}
	\caption{Agent event trigger situation}
	\label{fig3}
\end{figure}
\begin{figure}[htbp]
	\centering
	\includegraphics[width=0.9\linewidth]{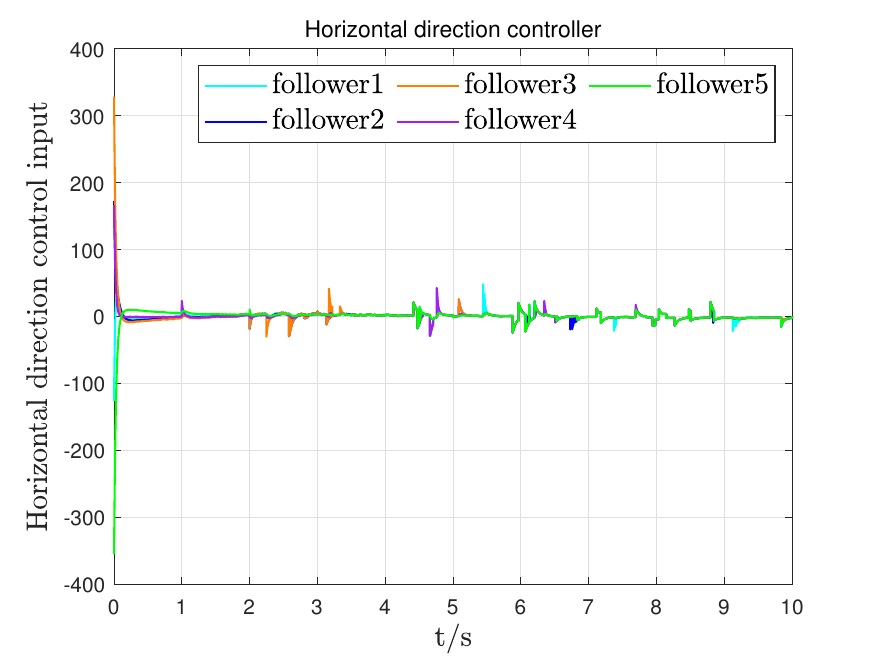}
	\caption{Control input of the follower agents in horizontal direction}
	\label{fig4}
\end{figure}
\begin{figure}[htbp]
	\centering
	\includegraphics[width=0.9\linewidth]{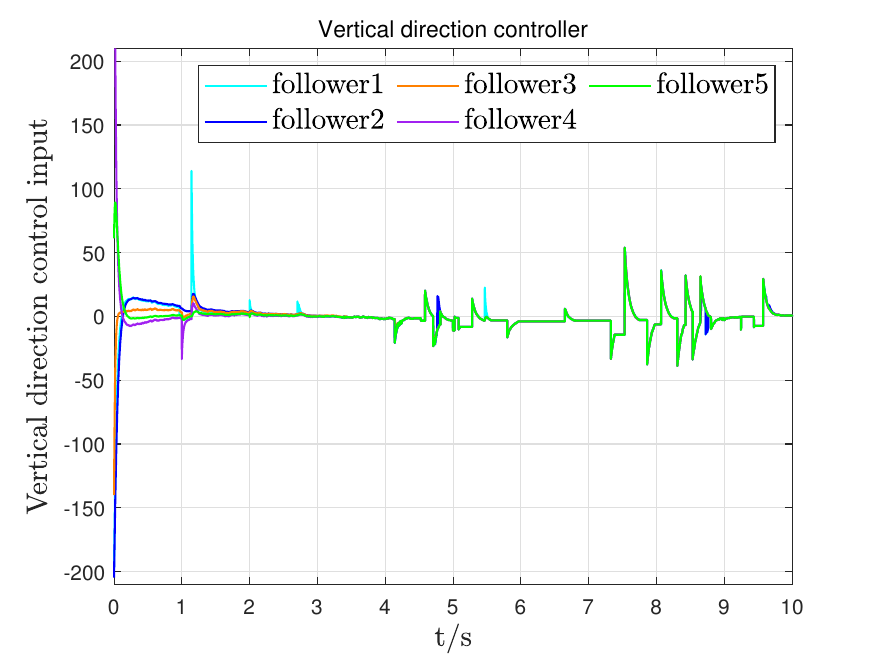}
	\caption{Control input of the follower agents in vertical direction}
	\label{fig5}
\end{figure}
\begin{figure}[htbp]
	\centering
	\includegraphics[width=0.9\linewidth]{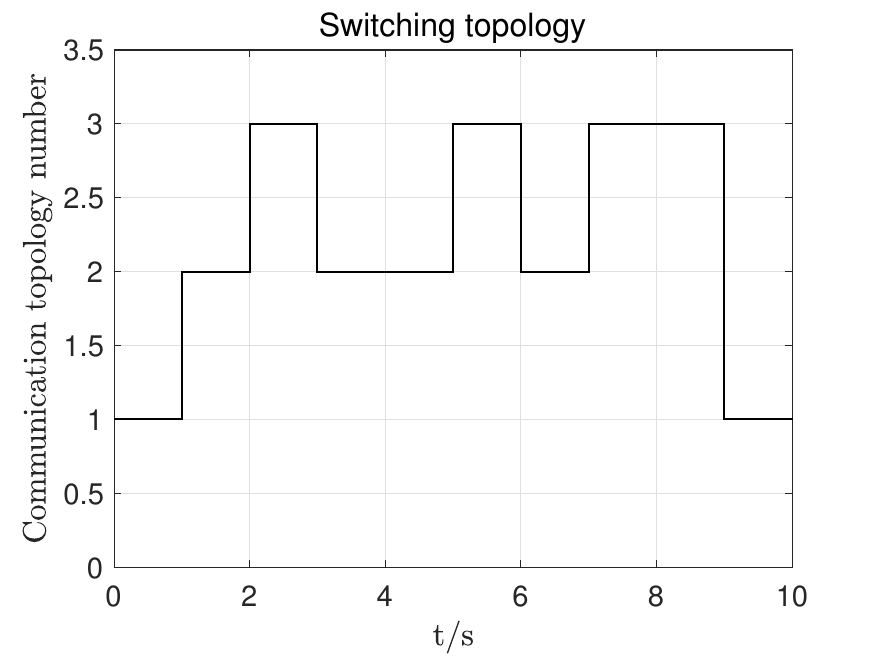}
	\caption{Communication topology network number}
	\label{fig6}
\end{figure}
\begin{figure}[htbp]
	\centering
	\includegraphics[width=0.9\linewidth]{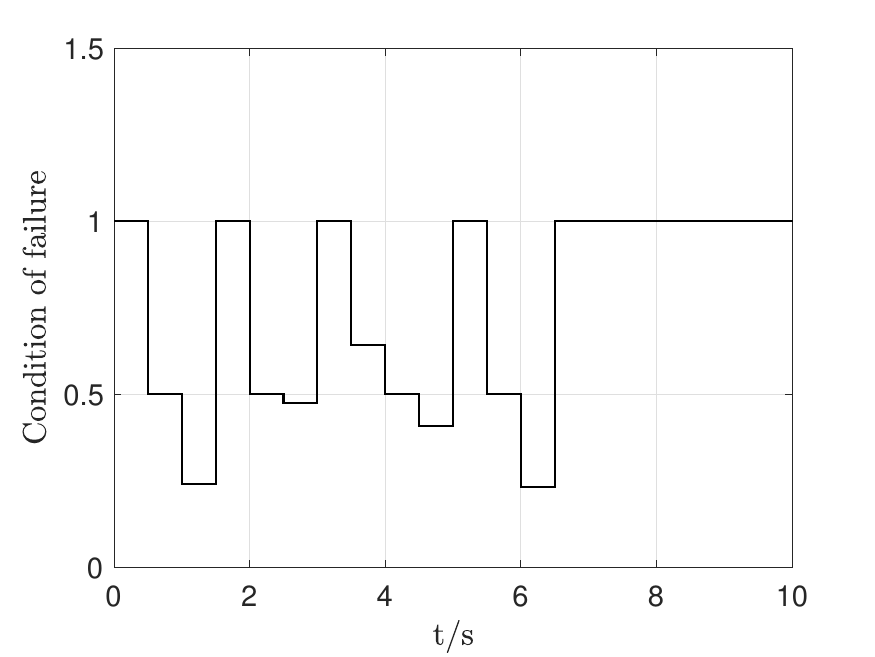}
	\caption{Actuator fault condition}
	\label{fig7}
\end{figure}
\begin{figure}[htbp]
	\centering
	\includegraphics[width=0.9\linewidth]{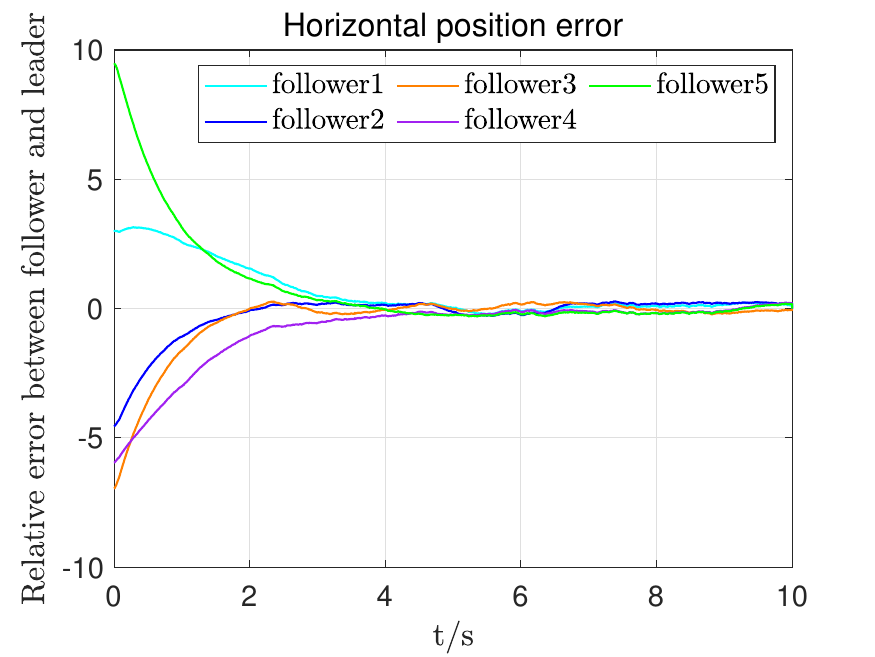}
	\caption{Formation tracking error in horizontal direction}
	\label{fig8}
\end{figure}
\begin{figure}[htbp]
	\centering
	\includegraphics[width=0.9\linewidth]{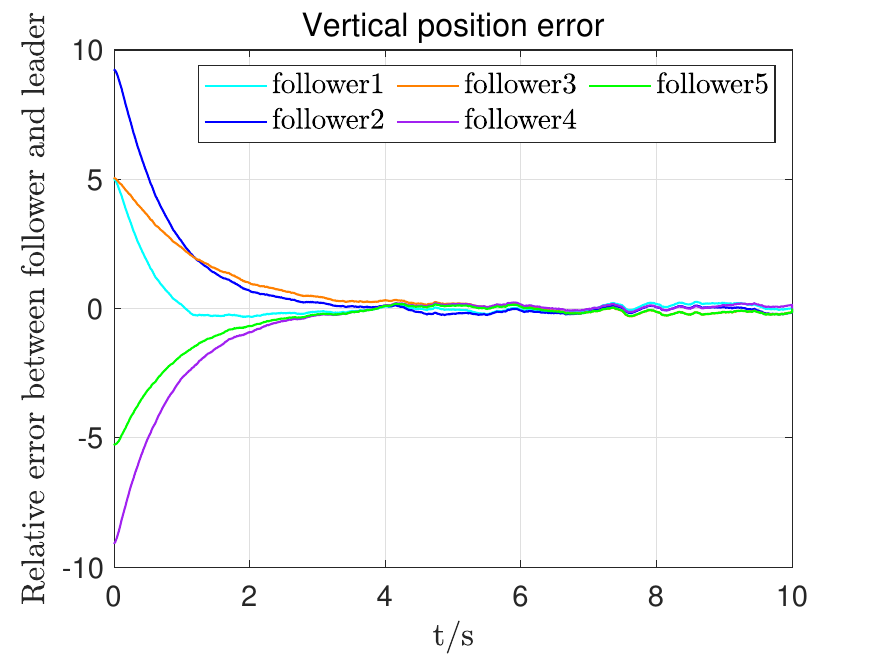}
	\caption{Formation tracking error in vertical direction}
	\label{fig9}
\end{figure}
As shown in Figure 1, the communication topology in this chapter satisfies Markov random switching and is randomly selected among the three topologies. And all communication topologies are directed connections. Figure 2 shows the relative position relationship between the leader and the follower. From Figure 2, it can be clearly seen that the follower keeps a pentagonal circular formation tracking around the leader located in the center of the pentagon. When the multi-agent system randomly switches topology, the agents do not arrive at the predetermined position and do not form the predetermined formation. After about 5s, the preset multi-agent formation formation is formed. Figure 3 shows the situation of multi-agent system event triggering. In the initial stage of multi-agent system formation, the departure frequency is very frequent due to excessive position error. When the formation of multi-agent system is basically stable, the triggered frequency drops.Figure 4 and Figure 5 Controller output on the axis and axis, where the controller output is large at the beginning of the operation, indicating that the error value between the initial position and the preset position is too large, but after 2s the controller input shows an exponential decline, indicating that the formation is about to be formed. The sudden increase of the controller during the formation operation of the multi-agent formation system indicates that the accumulated position error under the event trigger exceeds the set threshold, and the data transmission needs to be re-performed and the controller corrects the error. Figure 6. It is a random communication topology switching process of multi-agent system within 10s, in which the communication topology between agents follows Markov random process. Figure 7 shows the random failure of the actuator, in which the probability of failure is not $100\%$. When the actuator fails, the actuator can only execute part of the control commands. The smaller the value in Figure 7, the more serious the fault of the actuator and the worse the execution effect of the controller. Figures 8 and 9 show the preset position tracking error of the follower relative to the leader along and along the axis. It can be seen from the figure that the position error between the follower and the leader will always fluctuate because the formation is affected by external interference and in order to reduce the communication bandwidth. However, the formation of the multi-agent system will always be maintained, and the multi-agent formation system has been formed around 4s.

\section{Conclusion}\label{sec_6}
For nonlinear multi-agent system formation control with communication delay in random topology, a fault-tolerant sliding mode formation control strategy in Markov random switching topology is proposed in this chapter. A new event-triggered sliding mode formation control strategy for multi-agent systems is designed by combining event-triggered and sliding mode control techniques, and a sampling event-triggered controller is introduced to reduce the influence of network delay on multi-agent random topology formation systems. By using Lyapunov stability theorem to analyze and prove the stability of the system, it is shown that the proposed control algorithm can successfully solve the formation control problem which has the requirements of random topology, network delay and fault tolerance. Simulation results verify the effectiveness of the proposed controller, and show that the multi-agent formation system can still achieve the expected performance under the condition of stochastic topology and network delay.
\bibliographystyle{elsarticle-num}
\bibliography{wileyNJD-AMA}

\end{document}